\documentclass[11pt]{article}

\usepackage{empheq}
\usepackage{amsmath,amsfonts,amssymb,amsbsy,amscd}
\usepackage{mathrsfs,array}
\usepackage{subcaption}
\usepackage[a4paper, total={6in, 8in}]{geometry}
\usepackage{amsthm}
\usepackage[numbers,sort&compress]{natbib}

\newtheorem{theorem}{Theorem}[section]
\newtheorem{lemma}[theorem]{Lemma}
\newtheorem{remark}[theorem]{Remark}%

\usepackage{tikz}
\usetikzlibrary{arrows.meta}
\usetikzlibrary{3d}
\usetikzlibrary{shadings}
\usetikzlibrary{math}
\usepackage{adjustbox}\usepackage{multirow}
\usepackage{upgreek}
\usetikzlibrary{intersections, arrows, automata, backgrounds, calendar, chains, matrix, mindmap, patterns, petri, shadows, shapes.geometric, shapes.misc, spy, trees, shapes}

\definecolor{darkgreen}{rgb}{0, 0.7, 0}
\definecolor{orange}{rgb}{0.98, 0.6, 0.01}
 	\definecolor{napiergreen}{rgb}{0.16, 0.5, 0.0}

\usepackage{upgreek}

\newcolumntype{R}[2]{%
    >{\adjustbox{angle=#1,lap=\width-(#2)}\bgroup}%
    l%
    <{\egroup}%
}
\makeatletter
\newcommand{\thickhline}{%
    \noalign {\ifnum 0=`}\fi \hrule height 1pt
    \futurelet \reserved@a \@xhline
}
\newcolumntype{"}{@{\hskip\tabcolsep\vrule width 1pt\hskip\tabcolsep}}
\makeatother

\usepackage{pifont}

\usepackage{hyperref}
\hypersetup{
    colorlinks=true,
    linkcolor=blue,
    urlcolor=red,
    linktoc=all,
    allcolors=blue
}

\usepackage[nameinlink,capitalise,noabbrev]{cleveref}
\makeatletter
\newcommand{\myref}[1]{\cref{#1}\mynameref{#1}{\csname r@#1\endcsname}}
\newcommand{\Myref}[1]{\Cref{#1}\mynameref{#1}{\csname r@#1\endcsname}}

 	\definecolor{burgundy}{rgb}{0.5, 0.0, 0.13}
 	\definecolor{napiergreen}{rgb}{0.16, 0.5, 0.0}
    \definecolor{lasallegreen}{rgb}{0.03, 0.47, 0.19}
    \definecolor{airforceblue}{rgb}{0.36, 0.54, 0.66}

\definecolor{red}{rgb}{0.0 0.0 0.0}
\definecolor{mycolor}{rgb}{1,0.2,0.3}
\definecolor{beaublue}{rgb}{0.74, 0.83, 0.9}
\definecolor{oceanboatblue}{rgb}{0.0, 0.47, 0.75}
\definecolor{lightblue}{rgb}{0.78, 0.95, 1.0}	
\definecolor{amber}{rgb}{1.0, 0.6, 0.1}
\definecolor{darkblue}{rgb}{0.0, 0.0, 0.55}
\definecolor{lightgreen2}{rgb}{0.898, 1.000, 0.835}
\definecolor{lightgray}{rgb}{0.9, 0.9, 0.9}
\definecolor{lightblue}{rgb}{0.80, 0.93, 0.95}
\definecolor{orangepeel}{rgb}{1.0, 0.85, 0.5}

\definecolor{amber3}{rgb}{1.0, 0.74, 0.5}
\definecolor{lightgreen3}{rgb}{0.949, 1.0, 0.918}
\definecolor{lightblue3}{rgb}{0.90, 0.97, 0.975}
\definecolor{orangepeel3}{rgb}{1.0, 0.93, 0.75}

\definecolor{azure}{rgb}{0.0, 0.5, 1.0}

\usepackage{accents}
\newlength{\dhatheight}
\newcommand{\doublehat}[1]{%
    \settoheight{\dhatheight}{\ensuremath{\hat{#1}}}%
    \addtolength{\dhatheight}{-0.20ex}%
    \hat{\vphantom{\rule{1pt}{\dhatheight}}%
    \smash{\hat{#1}}}}

\numberwithin{equation}{section}

\allowdisplaybreaks

\title{A unified framework for N-phase Navier-Stokes Cahn-Hilliard Allen-Cahn mixture models with non-matching densities}

\author{M.F.P. ten Eikelder\thanks{e-mail: \texttt{marco.eikelder@tu-darmstadt.de}}
}
\date{%
    $^\dag$Institute for Mechanics, Computational Mechanics Group, Technical University of Darmstadt
}

\def\B#1{\mbox{\boldmath{$#1$}}}

\newcommand{\divg}{{\rm div}}

\newcommand{\nn}{\nonumber}

\newcommand{\mA}{\alpha}
\newcommand{\mB}{{\beta}}
\newcommand{\mC}{{\gamma}}

\newcommand{\bu}{\mathbf{u}}

\newcommand{\bw}{\mathbf{w}}

\newcommand{\bv}{\mathbf{v}}

\newcommand{\bx}{\mathbf{x}}
\newcommand{\by}{\B{y}}

\newcommand{\bJ}{\mathbf{J}}
\newcommand{\bh}{\mathbf{h}}
\newcommand{\bH}{\mathbf{H}}
\newcommand{\bj}{\mathbf{j}}
\newcommand{\bomega}{\boldsymbol{\omega}}
\newcommand{\bchi}{\boldsymbol{\chi}}
\newcommand{\bnu}{\boldsymbol{\nu}}
\newcommand{\trho}{\tilde{\rho}}

\newcommand{\bpi}{\boldsymbol{\pi}}

\def\be{\begin{equation}}
\def\ee{\end{equation}}
\def\ba{\begin{array}}
\def\ea{\end{array}}
\def\bea{\begin{eqnarray}}
\def\eea{\end{eqnarray}}
\def\beas{\begin{eqnarray*}}
\def\eeas{\end{eqnarray*}}
\newcommand{\bseq}{\begin{subequations}}
\newcommand{\eseq}{\end{subequations}}

\begin{document}

\maketitle

\begin{abstract}
Over the past few decades, numerous N-phase incompressible diffuse-interface flow models with non-matching densities have been proposed. Despite aiming to describe the same physics, these models are generally distinct, and an overarching modeling framework is absent.
This paper provides a unified framework for N-phase incompressible Navier-Stokes Cahn-Hilliard Allen-Cahn mixture models with a single momentum equation. The framework naturally emerges from continuum mixture theory, exhibits an energy-dissipative structure, and is invariant to the choice of fundamental variables. This opens the door to exploring connections between existing N-phase models and facilitates the computation of N-phase flow models rooted in continuum mixture theory.
\end{abstract}




\section{Introduction}

\subsection{Background}
Incompressible multi-phase flows are ubiquitous in nature, science and engineering, with a wide range of {\color{red}applications}\footnote{In this work, the term ``phase" denotes the different fluid materials/constituents (e.g. air and water).}. The development of continuum models (and corresponding methods) that describe these flows has been an active field of research for the last few  decades. This research can be (roughly) divided into (i) sharp interface models \cite{Hirt_Nichols_81,sethian2001evolution,ten2021novel,bothe2022sharp}, and (ii) diffuse-interface models. Within the diffuse-interface category, phase-field models constitute a well-known class \citep{cahn1958free1,cahn1959free2,anderson1998diffuse,gomez2018computational}. While we acknowledge the importance of each of the above approaches, the current article focuses on phase-field models.

Phase-field models have gained popularity over the last decades, and have become a versatile modeling technology with a wide range of applications in science and engineering. They offer resolutions to challenging moving boundary problems by simultaneously addressing the geometrical representation and the physical model, see e.g. \cite{anderson1998diffuse,steinbach2009phase}. By representing interfaces implicitly through continuous field variables, phase-field models eliminate the need for explicit boundary tracking, enabling accurate and efficient simulations of phenomena such as solidification \citep{boettinger2002phase}, crack propagation in fracture mechanics \citep{ambati2015phase}, and two-fluid flow dynamics \citep{yue2004diffuse}. 

The vast majority of incompressible, viscous, multi-phase flow models in the literature is restricted to two fluids. In the realm of phase-field modeling, a prototypical model is the Navier-Stokes Cahn-Hilliard Allen-Cahn (NSCHAC) model. The first model of this kind, now known as \textit{model H}, was proposed in \citet{hohenberg1977theory}. This model may be understood as a simplification of the more complete two-phase NSCHAC model in the sense that (i) it is restricted to matching fluid densities, and (ii) it does not permit mass transfer between phases (i.e. it does not contain an Allen-Cahn type term). The foundation of this model is largely based on empirical arguments; a derivation based on the concept of microforces (see \cite{gurtin1996generalized}) was established in \citet{gurtinmodel}. In subsequent years, several efforts have been made to relax the matching-density restriction, see e.g. \cite{lowengrub1998quasi,abels2012thermodynamically,aki2014quasi}, and see e.g. \cite{kay2007efficient,guo2014numerical,khanwale2022fully,ten2024divergence} for numerical simulations. Initially, these models were classified into two \textit{distinct} categories: (i) models with a mass-averaged mixture velocity, and (ii) models with a volume-averaged mixture velocity. In a recent article, we proposed a unified framework, rooted in continuum mixture theory, which leads to a single Navier-Stokes Cahn-Hilliard (NSCH) model that is invariant to the set of fundamental variables \citep{eikelder2023unified}; see \cite{ten2024divergence} for a divergence-conforming discretization with benchmarks. Contrary to the above-mentioned classification, the framework indicates that aforementioned classes of models coincide, up to minor modifications.

Although most research in the field of multi-phase flows focuses on $N=2$ phases, there are various $N$-phase ($N>2$) incompressible flow models. Similar to the two-phase case, the literature on $N$-phase models that (partly) utilize continuum mixture theory is divided into two categories: (i) models with a mass-averaged mixture velocity, and (ii) models with a volume-averaged mixture velocity. Without attempting to be complete, we mention the $N$-phase mass-averaged velocity models \cite{kim2005phase} ($N=3$) and \cite{heida2012development,li2014class} ($N\geq 2$), and the $N$-phase volume-averaged models \cite{dong2015physical,dong2018multiphase,huang2021consistent} ($N\geq 2$). Furthermore, there are incompressible $N$-phase NSCH models that are not (partly) based on continuum mixture theory, rather these models are established via coupling a multi-phase Cahn-Hilliard (CH) model to the {\color{red}Navier-Stokes equations, see \cite{boyer2006study,boyer2010cahn,toth2016phase,zhang2016phase,nurnberg2017numerical,xia2022modeling,xiao2024reduction}}. We also refer to several theoretical considerations of Allen-Cahn/Cahn-Hilliard (AC/CH) systems in isolation (ignoring inertial phenomena present in fluid mechanic systems), see e.g. \cite{eyre1993systems,boyer2014hierarchy,toth2015consistent,li2016multi,wu2017multiphase}, {\color{red}and to phase-field $N$-phase flow models \cite{xia2023conservative,mirjalili2024conservative} that are not of NSCH type}.

{\color{red}Although various $N$-phase models have been proposed, their differences in assumptions and methodologies pose challenges for both theoretical analysis and practical application. A unified perspective remains elusive, complicating efforts to compare and refine these models.}

\subsection{Objective and main results}
A number of the existing $N$-phase phase-field models, mentioned above, and in the references therein, provide different models (alongside with computational methodologies) for the \textit{same physical situation}: the dynamics of viscous, incompressible (isothermal) $N$-phase mixture flows. Naturally, there is some leeway in constitutive modeling, and not all models have the same complexity level\footnote{To organize the various existing models one can adopt the classification introduced in \cite{hutter2013continuum}. This classification is for example utilized in \cite{bothe2015continuum,hutter2018thermodynamics}.}. However, one can infer that models within the same complexity class are already distinct before constitutive modeling. The above observations raise questions regarding differences and connections between the models. While the aforementioned unified framework of NSCHAC models \citep{eikelder2023unified} is presented for the two-phase case, the adopted modeling principles therein are at the core not restricted to two phases. There are however a number of non-trivial considerations that come into play when examining the more general case $N\geq 2$. Important elements to consider are (i) symmetry properties with respect to the numbering of the phases, (ii) the reduction-consistency property (an $N$-phase system reduces to an $(N-M)$-phase system in absence of $M$ phases), and (iii) and the saturation constraint (volume-fractions/concentrations add up to one).

{\color{red}In light of these challenges, a systematic approach is needed to reconcile and unify existing models while addressing key theoretical considerations such as symmetry, reduction-consistency, and the saturation constraint. For this purpose we utilize continuum mixture theory \citep{truesdell1960classical} as point of departure. Continuum mixture theory provides a macroscopic framework for modeling systems composed of multiple interacting constituents, such as phases or chemical species. In this theory, each constituent is treated as a continuous field, characterized by its own set of properties, such as mass density, velocity, and concentration. These fields coexist and interact within the same spatial domain, governed by balance laws for mass, momentum, and energy. A key aspect of this mixture theory is its ability to account for inter-constituent interactions through constitutive relations, ensuring that the overall behavior reflects the combined effects of the individual phases. The framework serves as a foundation for deriving governing equations for multiphase flows and provides a systematic approach to connect microscopic processes with macroscopic behavior.}

The primary objective of this article is to lay down a unified framework of $N$-phase NSCHAC mixture models. We limit our focus to isothermal phases. In particular, we derive the following multi-phase-field model for phases (constituents) $\mA = 1, ..., N$:
\begin{subequations}\label{eq: intro mass}
  \begin{align}
   \partial_t (\rho \bv) + {\rm div} \left( \rho \bv\otimes \bv \right) + \sum_\mB \phi_\mB \nabla (\mu_\mB + \lambda)
    - {\rm div} \left(   \nu (2 \nabla^s \bv+\bar{\lambda}({\rm div}\bv) \mathbf{I}) \right)-\rho\mathbf{b} &=~ 0, \label{eq: intro mass: mom}\\
  \partial_t \phi_\mA  + {\rm div}(\phi_\mA  \bv) +\rho_\mA^{-1}{\rm div} (\hat{\bJ}_\mA + \hat{\bj}_\mA )  -\rho_\mA^{-1} \hat{\zeta}_\mA&=~0,\label{eq: intro mass: mass}\\
  \hat{\bJ}_\mA + \sum_\mB \mathbf{M}_{\mA\mB}\nabla g_\mB&=~0, \label{eq: intro mass: bJ}\\
  \hat{\bj}_\mA + \sum_\mB \mathbf{K}_{\mA\mB}\nabla g_\mB&=~0, \label{eq: intro mass: bj}\\
  \hat{\zeta}_\mA + \sum_\mB m_{\mA\mB} g_\mB&=~0, \label{eq: intro mass: zeta}
  \end{align}
\end{subequations}
{\color{red}for $\mA = 1, \ldots, N$, subject to $\sum_\mB \phi_\mB =1$, where $\phi_\mA$ is the volume fraction of constituent $\mA$, $\mathbf{v}$ denotes the fluid velocity, $\rho_\mA$ and $\trho_\mA$ represent the constituent mass densities, $\rho = \sum_\mB \trho_\mB$ is the mixture density, $\mathbf{b}$ is the force vector, $\nu$ is the dynamic viscosity, $\nu \bar{\lambda}$ is the second viscosity coefficient, $\nabla^s \bv$ represents the symmetric velocity gradient, and $\lambda$ is the Lagrange multiplier pressure.} Additionally, $\mu_\alpha, g_\mA$ are constituent chemical potentials, and $\mathbf{M}_{\mA\mB}, \mathbf{K}_{\mA\mB}$ and $m_{\mA\mB}$ are mobility parameters. The model is composed of equation \eqref{eq: intro mass: mom} that details the mixture momentum equation, $N$ constituent mass balance equations \eqref{eq: intro mass: mass}, and models for peculiar velocities \eqref{eq: intro mass: bJ}, and conservative and non-conservative mass transfer models \eqref{eq: intro mass: bj}-\eqref{eq: intro mass: zeta}. Model \eqref{eq: intro mass} is expressed in terms of the mass-averaged mixture velocity $\mathbf{v}$; an alternative -- but equivalent -- formulation emerges when adopting the volume-averaged mixture velocity $\mathbf{u}$:
\begin{subequations}\label{eq: intro volume}
  \begin{align}
   \partial_t \left(\rho   \bv \right)  + {\rm div} \left( \rho \bv \otimes \bv \right)  + \sum_\mB \phi_\mB \nabla (\mu_\mB + \lambda) & \nn\\
    - {\rm div} \left(   \nu \left(2\nabla^s \bv+\bar{\lambda}{\rm div}\bv \mathbf{I}\right) \right) -\rho\mathbf{b} &=~ 0, \label{eq: intro volume: mom}\\
     {\rm div} \bu - \sum_{\mB} \rho_\mB^{-1}(\hat{\zeta}_\mB + {\rm div} \hat{\mathbf{j}}_\mB) &=~0,\\
  \partial_t \phi_\mA  + {\rm div}\left(\phi_\mA  \bv \right) +\rho_\mA^{-1}{\rm div} (\hat{\bJ}_\mA + \hat{\bj}_\mA )  - \rho_\mA^{-1} \hat{\zeta}_\mA &=~0,
  \end{align}
\end{subequations}
for $\mA = 1, \ldots, N-1$, subject to $\sum_\mB \phi_\mB =1$ with $\bv = \bu- \sum_\mB \rho_\mB^{-1} \hat{\bJ}_\mB $, where $\hat{\bJ}_\mA, \hat{\bj}_\mA$ and $\hat{\zeta}_\mA$ are defined in \eqref{eq: intro mass: bJ}, \eqref{eq: intro mass: bj} and \eqref{eq: intro mass: zeta}, respectively. Analogously to the above formulation, the model is comprised of a mixture momentum equation \eqref{eq: intro mass: mom}, and $N$ constituent mass balance laws \eqref{eq: intro mass: mass}. We provide precise definitions of all quantities in the remainder of the article. {\color{red}A key property of the framework is its invariance to the set of fundamental variables, both before and after constitutive modeling (see \cref{fig: invariance}).}

\begin{figure}
\resizebox{0.975\textwidth}{!}{\begin{tikzpicture}

\path[fill=lightblue, draw=black, rounded corners=15pt] (-4, 3.5) rectangle (-6, 4.5);
\node[text=black, font=\rmfamily\bfseries\scriptsize, text centered] at (-5, 4.125) {Bal. Laws};
\node[text=black, font=\rmfamily\small, text centered] at (-5, 3.75) {$\bv, \left\{\phi_\mA\right\}$};
\path[fill=lightblue, draw=black, rounded corners=15pt] (-4, 3.5-3) rectangle (-6, 4.5-3);
\node[text=black, font=\rmfamily\bfseries\scriptsize, text centered] at (-5, 4.125-3) {Mix. Model};
\node[text=black, font=\rmfamily\small, text centered] at (-5, 3.75-3) {$\bv, \left\{\phi_\mA\right\}$};

\path[fill=lightgreen2, draw=black, rounded corners=15pt] (0, 3.5) rectangle (-2, 4.5);
\node[text=black, font=\rmfamily\bfseries\scriptsize, text centered] at (-1, 4.125) {Bal. Laws};
\node[text=black, font=\rmfamily\small, text centered] at (-1, 3.75)  {$\bu, \left\{\phi_\mA\right\}$};
\path[fill=lightgreen2, draw=black, rounded corners=15pt] (0, 3.5-3) rectangle (-2, 4.5-3);
\node[text=black, font=\rmfamily\bfseries\scriptsize, text centered] at (-1, 4.125-3) {Mix. Model};
\node[text=black, font=\rmfamily\small, text centered] at (-1, 3.75-3)  {$\bu, \left\{\phi_\mA\right\}$};

\path[fill=orangepeel, draw=black, rounded corners=15pt] (4, 3.5) rectangle (2, 4.5);
\node[text=black, font=\rmfamily\bfseries\scriptsize, text centered] at (3, 4.125) {Bal. Laws};
\node[text=black, font=\rmfamily\small, text centered] at (3, 3.75) {$\bv, \left\{c_\mA\right\}$};
\path[fill=orangepeel, draw=black, rounded corners=15pt] (4, 3.5-3) rectangle (2, 4.5-3);
\node[text=black, font=\rmfamily\bfseries\scriptsize, text centered] at (3, 4.125-3) {Mix. Model};
\node[text=black, font=\rmfamily\small, text centered] at (3, 3.75-3) {$\bv, \left\{c_\mA\right\}$};

\path[fill=amber, draw=black, rounded corners=15pt] (8, 3.5) rectangle (6, 4.5);
\node[text=black, font=\rmfamily\bfseries\scriptsize, text centered] at (7, 4.125) {Bal. Laws};
\node[text=black, font=\rmfamily\small, text centered] at (7, 3.75)  {$\bu, \left\{c_\mA\right\}$};
\path[fill=amber, draw=black, rounded corners=15pt] (8, 3.5-3) rectangle (6, 4.5-3);
\node[text=black, font=\rmfamily\bfseries\scriptsize, text centered] at (7, 4.125-3) {Mix. Model};
\node[text=black, font=\rmfamily\small, text centered] at (7, 3.75-3)  {$\bu, \left\{c_\mA\right\}$};

\node[text=black, font=\rmfamily\small, text centered] at (-6, 2.45)  {closure};
\node[text=black, font=\rmfamily\small, text centered] at (-2, 2.45)  {closure};
\node[text=black, font=\rmfamily\small, text centered] at (2, 2.45) 
 {closure};
\node[text=black, font=\rmfamily\small, text centered] at (6, 2.45) 
 {closure};

\node[text=black, font=\rmfamily\small, text centered] at (-3.0, 3.25)  {transform};
\node[text=black, font=\rmfamily\small, text centered] at (-3.0, 1.75)  {transform};
\node[text=black, font=\rmfamily\small, text centered] at (1.0, 3.25)  {transform};
\node[text=black, font=\rmfamily\small, text centered] at (1.0, 1.75)  {transform};
\node[text=black, font=\rmfamily\small, text centered] at (5.0, 3.25)  {transform};
\node[text=black, font=\rmfamily\small, text centered] at (5.0, 1.75)  {transform};

\draw[->, thick, bend left, stealth-] (-5, 1.5) to (-5, 3.5);

\draw[->, thick, bend left, stealth-] (-1, 1.5) to (-1, 3.5);

\draw[->, thick, bend left, stealth-] (3, 1.5) to (3, 3.5);

\draw[->, thick, bend left, stealth-] (7, 1.5) to (7, 3.5);

\draw[->, thick, bend left, stealth-] (-2, 3.8) to (-4, 3.8);
\draw[->, thick, bend left, stealth-] (-4, 4.15) to (-2, 4.15);

\draw[->, thick, bend left, stealth-] (-2, 0.75) to (-4, 0.75);
\draw[->, thick, bend left, stealth-] (-4, 1.1)  to (-2, 1.1);

\draw[->, thick, bend left, stealth-] (2, 3.8)  to (0, 3.8);
\draw[->, thick, bend left, stealth-] (0, 4.15) to (2, 4.15);

\draw[->, thick, bend left, stealth-] (2, 0.75) to (0, 0.75);
\draw[->, thick, bend left, stealth-] (0, 1.1)  to (2, 1.1);

\draw[->, thick, bend left, stealth-] (6, 3.8)  to (4, 3.8);
\draw[->, thick, bend left, stealth-] (4, 4.15) to (6, 4.15);

\draw[->, thick, bend left, stealth-] (6, 0.75) to (4, 0.75);
\draw[->, thick, bend left, stealth-] (4, 1.1)  to (6, 1.1);

\end{tikzpicture}}
    \caption{{\color{red}Invariance of the unified framework, both at the level of balance laws (Bal. Laws) and, after closure, at the level of mixture models (Mix. Model).}}
    \label{fig: invariance}
\end{figure}

The classification as an NSCHAC model is evident in the combination of a momentum equation with ($N$) mass balance laws that are of Cahn-Hilliard Allen-Cahn type for specific free energy choices. The Cahn-Hilliard components appear in the third members of the mass balance laws, whereas the Allen-Cahn character materializes in the latter terms of the mass balance laws. Furthermore, the model -- in both formulations -- displays a strong coupling between the various equations; through the constituent densities $\trho_\mA$, the velocity $\bv$ (or $\bu$) and the Lagrange multiplier pressure $\lambda$.

The secondary objective of this article is reveal  connections between model \eqref{eq: intro mass}-\eqref{eq: intro volume} and existing models in the literature. First, we compare model \eqref{eq: intro mass}-\eqref{eq: intro volume} with the unified NSCHAC model \citep{eikelder2023unified} for the situation of two phases. Subsequently, we compare the framework to that of \cite{dong2018multiphase}. Finally, we discuss the connections of the proposed framework with the mixture-theory-compatible $N$-phase model \citep{eikelder2023thermodynamically}.

\subsection{Plan of the paper}
The remainder of the paper is organized as follows. In \cref{sec: mix theory} we present the continuum theory of rational mechanics for incompressible isothermal fluid mixtures, highlighting the connections between different quantities and formulations of evolution equations. Next, in \cref{sec: 2nd law}, we conduct constitutive modeling using the Coleman-Noll procedure. Following that, \cref{sec: Properties} addresses the properties of the model. Subsequently, in \cref{sec: Connections} we explore the connections of the novel model with existing models in the literature. Finally, in \cref{sec: discussion}, we provide a conclusion and outlook.

\section{Continuum mixture theory}\label{sec: mix theory}

The purpose of this section is to outline the continuum theory of mixtures for incompressible constituents, excluding thermal effects. This section aligns with \cite{eikelder2023thermodynamically} at several points.

The continuum theory of mixtures is grounded in three general principles introduced in the pioneering work of \cite{truesdell1960classical}:
\begin{enumerate}
    \item \textit{All properties of the mixture must be mathematical consequences of properties of the constituents.}
\item \textit{So as to describe the motion of a constituent, we may in imagination isolate it from the rest of the mixture, provided we allow properly for the actions of the other constituents upon it.}
\item \textit{The motion of the mixture is governed by the same equations as is a single body.}
\end{enumerate}
The first principle communicates that the mixture is made up of its constituent parts. The second principle asserts the connection of the different components of the physical model through interaction terms. Lastly, the latter principle states that one can not distinguish the motion of a mixture from that of a single fluid.

In \cref{sec: prelim} we introduce the fundamentals of the continuum theory of mixtures and the necessary kinematics. Then, in \cref{sec: const BL,sec: mixture BL}, we provide balance laws of individual constituents and associated mixtures.

\subsection{Preliminaries}\label{sec: prelim}

In the continuum theory of mixtures the material body $\mathscr{B}$ is comprised of $N$ constituent bodies $\mathscr{B}_\mA$, with $\mA = 1, \dots, N$. The bodies $\mathscr{B}_\mA$ are permitted to simultaneously occupy a shared region in space. Denoting by $\mathbf{X}_{\mA}$ the spatial position of a particle of $\mathscr{B}_\mA$ in the Lagrangian (reference) configuration, the (invertible) deformation map defines the spatial position of a particle:
\begin{align}
    \mathbf{x} := \bchi_{\mA}(\mathbf{X}_{\mA},t), 
\end{align}
where $\mathbf{x} \in \Omega$, with $\Omega \in \mathbb{R}^d$ the domain (dimension $d$). {\color{red}We refer for more details on continuum mixture theory to \cite{truesdell1960classical}, and sketch the situation in \cref{fig:cont mix}.}
\begin{figure}
\resizebox{0.95\textwidth}{!}{\begin{tikzpicture}[>=stealth,
axis/.style={densely dashed,font=\small}]

\pgfdeclareradialshading{ballshading}{
 \pgfpoint{-10bp}{10bp}}
 {color(0bp)=(gray!30!white); 
  color(9bp)=(gray!55!white);
  color(18bp)=(gray!75!white); 
  color(25bp)=(gray!70!black); 
  color(50bp)=(gray!50!black)}

\coordinate (K) at (3,1);
\shade[shading=ballshading] (K) plot [smooth cycle,tension=0.7] coordinates {(3,1) (5,1.2) (7,1) (8,3) (7,4.5) (5,4.5) (2,4) (1.7,2.5)};

\coordinate (K) at (12,1);
\shade[shading=ballshading] (K) plot [smooth cycle,tension=0.7] coordinates {(12,1-1)  (14,1-1) (14,2.75-1) (13,2.5-1) (11.5,2.5-1)};

\coordinate (K) at (12,1);
\shade[shading=ballshading] (K) plot [smooth cycle,tension=0.7] coordinates {(12,1+3)  (14.5,1+3) (14,2.0+3) (13,2.5+3) (11.5,2.5+3)};

\draw[->, thick, bend right, stealth-] (5.1,2.8) to  (12.8,0.9);
\draw[->, thick, bend left, stealth-] (5.1,3.2) to  (12.8,4.6);
\node[text=black, font=\large\rmfamily\bfseries, text centered] at (13.3,1) {$\bullet\mathbf{X}_N$};
\node[text=black, font=\large\rmfamily\bfseries, text centered] at (13.3,4.5) {$\bullet\mathbf{X}_1$};
\node[text=black, font=\rmfamily\bfseries, text centered] at (13.0,2.4) {$\bullet$};
\node[text=black, font=\rmfamily\bfseries, text centered] at (13.0,2.65) {$\bullet$};
\node[text=black, font=\rmfamily\bfseries, text centered] at (13.0,2.9) {$\bullet$};
\node[text=black, font=\large\rmfamily\bfseries, text centered] at (4.8,3.0) {$\mathbf{x}~~\bullet$};
\node[text=black, font=\large\rmfamily\bfseries, text centered] at (9.3,1) {$\boldsymbol{\chi}_N$};
\node[text=black, font=\large\rmfamily\bfseries, text centered] at (9.3,4.8) {$\boldsymbol{\chi}_1$};
\node[text=black, font=\large\rmfamily\bfseries, text centered] at (14,5.5) {$\mathscr{B}_1$};
\node[text=black, font=\large\rmfamily\bfseries, text centered] at (14,2.1) {$\mathscr{B}_N$};
\node[text=black, font=\large\rmfamily\bfseries, text centered] at (5,4.8) {$\mathscr{B}$};

\end{tikzpicture}}
    \caption{{\color{red}Situation sketch continuum mixture theory.}}
    \label{fig:cont mix}
\end{figure}
We introduce the constituent partial mass density $\tilde{\rho}_{\mA}$ and specific mass density $\rho_{\mA}>0$ respectively as:
\begin{subequations}\label{eq: def trhoA rhoA}
  \begin{align}
  \tilde{\rho}_{\mA}(\bx,t) :=&~ \displaystyle\lim_{ \vert V \vert \rightarrow 0} \dfrac{M_{\mA}(V)}{\vert V \vert},\\
  \rho_{\mA}(\bx,t) :=&~ \displaystyle\lim_{\vert V_{\mA}\vert  \rightarrow 0} \dfrac{M_{\mA}(V)}{\vert V_{\mA}\vert },
\end{align}
\end{subequations}
where $V \subset \Omega$ (measure $\vert V \vert$)
is an arbitrary control volume around $\mathbf{x}$, $V_{\mA} \subset V$ (measure $\vert V_{\mA}\vert$) is the volume of constituent $\mA$ so that  $V =\cup_{\mA}V_{\mA}$. Here, the constituents masses are $M_{\mA}=M_{\mA}(V)$, and the total mass in $V$ is $M=M(V)=\sum_{\mA}M_{\mA}(V)$. The mixture density is the sum of the partial mass densities:
\begin{align}\label{eq: def rho}
\rho(\bx,t):=&~ \displaystyle\lim_{ \vert V \vert \rightarrow 0} \dfrac{M(V)}{\vert V \vert}=\displaystyle\sum_{\mA}\trho_\mA(\bx,t).
\end{align}
Additionally, we introduce the mass concentrations (or mass fractions) and volume fractions respectively as:
\begin{subequations}\label{eq: def c phi alpha}
  \begin{align}
    c_\mA(\bx,t):=&~ \displaystyle\lim_{ \vert V \vert \rightarrow 0} \dfrac{M_{\mA}(V)}{M(V)}=\dfrac{\trho_\mA}{\rho},\\
    \phi_\mA(\bx,t):=&~ \displaystyle\lim_{ \vert V \vert \rightarrow 0} \dfrac{|V_{\mA}|}{|V|}=\dfrac{\trho_\mA}{\rho_\mA}, \label{eq: phi trho trho}
  \end{align}
\end{subequations}
which sum up to one:
\begin{subequations}\label{eq: sum c phi}
  \begin{align}
    \displaystyle\sum_{\mA} c_{\mA}(\bx,t)=&~ 1, \label{eq: sum c}\\
    \displaystyle\sum_{\mA} \phi_{\mA}(\bx,t)=&~ 1. \label{eq: sum phi}
  \end{align}
\end{subequations}
We assume that the constituents are incompressible, meaning that the specific mass densities are (constituent-wise) constant:
\begin{align}\label{eq: incompressibility}
    \rho_\mA(\bx,t) = \rho_\mA.
\end{align}
By means of the incompressibility of the constituents, \eqref{eq: incompressibility}, and the definitions \eqref{eq: def c phi alpha}, the volume fractions and concentrations are related by:
\begin{subequations}\label{eq: relation c phi}
    \begin{align}
        \phi_\mA =&~ \frac{c_\mA}{\rho_\mA}\left(\sum_\mB\frac{c_\mB}{\rho_\mB}\right)^{-1},\\
        c_\mA =&~ \rho_\mA\phi_\mA \left(\sum_\mB\rho_\mB\phi_\mB\right)^{-1},
    \end{align}
\end{subequations}
for $\mA =1, ..., N$.

\begin{remark}[Incompressibility $N$-phase model]
The relations \eqref{eq: relation c phi} hinge on the assumption that the constituents are incompressible, definition \eqref{eq: incompressibility}. The variables $\phi_\mA$ (or $c_\mA$) are interdependent via \eqref{eq: sum c phi}, which must be explicitly considered when formulating or deducing relationships to avoid overdetermined or inconsistent expressions. For example, the mappings \eqref{eq: relation c phi} are not invertible. We discuss these challenges throughout the article, and in \cref{appendix: sec: Equivalence of modeling restrictions}.
\end{remark}

\begin{remark}[Alternative definitions incompressible mixtures]
  Besides the current definition of incompressible constituents \eqref{eq: incompressibility}, which is frequently adopted in the literature (see e.g. \cite{li2014class,dong2015physical,dong2018multiphase,huang2021consistent}), there exist other notions of incompressibility in mixture flows. We refer for an alternative to \cite{bothe2023multicomponent} and the references therein.
\end{remark}

We proceed with the introduction of the material time derivative $\grave{\uppsi}_\mA$ of the differentiable constituent function $\uppsi_\mA$:
\begin{align}
\grave{\uppsi}_\mA=\partial_t\uppsi_{\mA}(\mathbf{X}_{\mA},t) \vert_{\mathbf{X}_\mA}.
\end{align}
Here we adopt the notation $\vert_{\mathbf{X}_\mA}$ to indicate that $\mathbf{X}_\mA$ is held fixed. The constituent velocity now follows as the constituent material derivative of the deformation map:
\begin{align}
\bv_{\mA}(\mathbf{x},t)=\partial_t\bchi_{\mA}(\mathbf{X}_{\mA},t) \vert_{\mathbf{X}_\mA} = \grave{\bchi}_\mA.
\end{align}
In contrast to the mixture density, there appear various mixture velocities in the literature. Among the most popular ones are the mass-averaged velocity, denoted $\bv$, and the volume-averaged velocity, denoted $\bu$, which are respectively given by:
\begin{subequations}\label{eq: mix velo}
  \begin{align}
    \bv(\bx,t) =&~ \displaystyle\sum_{\mA} c_\mA(\bx,t) \bv_\mA(\bx,t),\\
    \bu(\bx,t) =&~ \displaystyle\sum_{\mA} \phi_\mA(\bx,t) \bv_\mA(\bx,t).
\end{align}
\end{subequations}
We introduce peculiar velocities of the constituents relative to both mixture velocities:
\begin{subequations}\label{eq: def bwj}
  \begin{align}
    \bw_{\mA}(\bx,t):=&~\bv_{\mA}(\bx,t)-\bv(\bx,t),\\
    \bomega_{\mA}(\bx,t):=&~\bv_{\mA}(\bx,t)-\bu(\bx,t).
  \end{align}
\end{subequations}
Additionally, we define the following (scaled) peculiar velocities (that depend on $\bx$ and $t$):
\begin{subequations}\label{eq: def J and h}
    \begin{align}
        \bJ_\mA :=&~ \trho_\mA \bw_\mA, \label{eq: def J h: J}\\
        \bh_\mA = &~ \phi_\mA \bw_\mA,\\
        \bJ_\mA^u = &~ \trho_\mA \bomega_\mA\\
        \bh_\mA^u = &~ \phi_\mA \bomega_\mA.\label{eq: def J h: hu}
    \end{align}
\end{subequations}
\begin{remark}[Terminology peculiar velocities]
  The quantities \eqref{eq: def bwj} and \eqref{eq: def J and h} are in the literature often referred to as ``diffusion velocities'' and ``diffusive fluxes'', respectively. This terminology is natural because the terms \eqref{eq: def J and h} appear in constituent mass balance laws (see \cref{sec: const BL}) as flux terms, and their constitutive models (see \cref{sec: const mod: subsec: select}) have a diffusive character. However, utilizing constitutive models for \eqref{eq: def J and h} is not essential (see \cite{eikelder2023thermodynamically}), and therefore we use the terminology ``(scaled) peculiar velocity'' to reflect their original definitions \eqref{eq: def bwj} and \eqref{eq: def J and h}.
\end{remark}
Direct consequences of \eqref{eq: def bwj}, \eqref{eq: def J h: J}, and \eqref{eq: def J h: hu} are the properties:
  \begin{subequations}\label{eq: rel gross motion zero}
  \begin{align}
    \displaystyle\sum_{\mA} \bJ_\mA  =&~ 0, \label{eq: rel gross motion zero: J}\\
    \displaystyle\sum_{\mA} \bh_\mA^u  =&~ 0.
  \end{align}
  \end{subequations}
The relation between the mass-averaged and volume-averaged velocities is specified in the following lemma.
\begin{lemma}[Relation mass-averaged and volume-averaged velocities]\label[lemma]{lem: mass av vs vol av}
The mass-averaged and volume-averaged velocity variables are related via:
\begin{subequations}
  \begin{align}
    \bu =&~ \bv + \displaystyle\sum_\mA \rho_\mA^{-1} \bJ_\mA = \bv + \sum_\mA \bh_\mA,\\
    \bv =&~ \bu + \rho^{-1} \sum_\mA \bJ_\mA^u.
  \end{align}
\end{subequations}
\end{lemma}
\begin{proof}
    These relations result from the sequences of identities:
    \begin{subequations}
      \begin{align}
        \bu = \sum_\mA \phi_\mA \bv_\mA= \sum_\mA \phi_\mA \bw_\mA + \sum_\mA \phi_\mA \bv = \sum_\mA \rho_\mA^{-1} \bJ_\mA +  \bv,\\
        0 = \sum_\mA \bJ_\mA = \sum_\mA \trho_\mA (\bv_\mA-\bv) = \sum_\mA \trho_\mA (\bv_\mA-\bu + \bu-\bv) \nn\\=  \sum_\mA \bJ_\mA^u + \rho(\bu-\bv).\\ \nn
      \end{align}
    \end{subequations}
\end{proof}
The relation between the scaled peculiar velocities is displayed in the next lemma.
\begin{lemma}[Relation scaled peculiar velocities]\label[lemma]{lem: scaled peculiar velo}
The scaled peculiar velocities are related via:
\begin{subequations}
  \begin{align}
        \bJ_\mA =&~  \bJ_\mA^u - c_\mA\sum_\mB \bJ_\mB^u,\\
        \bJ_\mA^u =&~ \bJ_\mA - \trho_\mA \sum_\mB \rho_\mB^{-1} \bJ_\mB,\\
        \bh_\mA =&~  \bh_\mA^u - \phi_\mA\rho^{-1} \sum_\mB \rho_\mB\bh_\mB^u,\\
        \bh_\mA^u =&~ \bh_\mA - \phi_\mA \sum_\mB \bh_\mB.
  \end{align}
\end{subequations}
\end{lemma}
\begin{proof}
    These identities are a direct consequence of \cref{lem: mass av vs vol av}.
\end{proof}

Lastly, we define the material derivative of the mixture relative to the mass-averaged velocity: 
\begin{align}\label{eq: mat der}
    \dot{\uppsi}(\bx,t) =&~ \partial_t \uppsi(\bx,t) + \bv(\bx,t)\cdot \nabla \uppsi(\bx,t).
\end{align}

\subsection{Constituent balance laws}\label{sec: const BL}

In the continuum theory of mixtures, each constituent  moves according to a distinct set of balance laws, as specified by the second general principle. These laws incorporate terms that model the interactions among the different constituents. The following local balance laws apply to the motion of each constituent $\mA = 1, \dots, N$ for all $\mathbf{x}\in \Omega$ and $t >0$:
\begin{subequations}\label{eq: BL const}
  \begin{align}
        \partial_t \trho_\mA + {\rm div}(\trho_\mA \bv_\mA) &=~ \gamma_\mA, \label{eq: local mass balance constituent j} \\
        \partial_t (\trho_\mA\bv_\mA) + {\rm div} \left( \trho_\mA\bv_\mA\otimes \bv_\mA \right) -  {\rm div} \mathbf{T}_\mA -  \trho_\mA \mathbf{b}_\mA &=~ \boldsymbol{\pi}_\mA,\label{eq: lin mom constituent j}\\
        \mathbf{T}_\mA-\mathbf{T}_\mA^T &=~\mathbf{N}_\mA.\label{eq: ang mom constituent j}
  \end{align}
\end{subequations}
Equations \eqref{eq: local mass balance constituent j} describe the local constituent mass balance laws, where the interaction terms $\gamma_{\mA}$ denote the mass supply of constituent $\mA$ due to chemical reactions with the other constituents. Then, \eqref{eq: lin mom constituent j} represent the local constituent linear momentum balance laws, where $\mathbf{T}_\mA$ is the Cauchy stress tensor of constituent $\mA$, $\mathbf{b}_\mA$ is the constituent external body force, and $\bpi_\mA$ is the momentum exchange rate of constituent $\mA$ with the other constituents. We assume equal body forces ($\mathbf{b}_\mA= \mathbf{b}$ for $\mA = 1, \dots, N$) throughout the article. Additionally, we restrict to gravitational body forces: $\mathbf{b} = -b \boldsymbol{\jmath} = -b \nabla y$, with $y$ being the vertical coordinate, $\boldsymbol{\jmath}$ the vertical unit vector, and $b$ a constant. Finally, \eqref{eq: ang mom constituent j} describes the local constituent angular momentum balance with $\mathbf{N}_\mA$ the intrinsic moment of momentum.

We introduce a split of the mass transfer term into a conservative part and a potentially non-conservative contribution via:
\begin{align}\label{eq: split gamma}
  \gamma_\mA = \zeta_\mA - {\rm div} \mathbf{j}_\mA.
\end{align}
The mass balance laws \eqref{eq: BL constitutive: mass} take the form:
\begin{align}
    \partial_t \trho_\mA + {\rm div}(\trho_\mA \bv_\mA) + {\rm div} \bj_\mA = \zeta_\mA.
\end{align}
By invoking the definitions in \cref{sec: prelim}, one can deduce various alternative -- equivalent -- formulations of the constituent mass balance laws \eqref{eq: local mass balance constituent j}, such as:
\begin{subequations}\label{eq: BL mass constituent material der}
  \begin{align}
    \partial_t \trho_\mA + {\rm div}(\trho_\mA \bv) + {\rm div} (\bJ_\mA + \bj_\mA) &=~ \zeta_\mA, \label{eq: BL mass constituent material der: 1}\\
    \partial_t \trho_\mA + {\rm div}(\trho_\mA \bu) + {\rm div} (\mathbf{J}_\mA^u  + \bj_\mA) &=~ \zeta_\mA, \label{eq: BL mass constituent material der: 2}\\
    \partial_t \phi_\mA + {\rm div}(\phi_\mA \bv) + {\rm div} \mathbf{h}_\mA + \rho_\mA^{-1}{\rm div} \bj_\mA &=~ \rho_\mA^{-1}\zeta_\mA,\label{eq: BL mass constituent material der: 3}\\      
    \partial_t \phi_\mA  + {\rm div}\left(\phi_\mA \bu\right)+ {\rm div}\bh_\mA^u + \rho_\mA^{-1}{\rm div} \bj_\mA &=~ \rho_\mA^{-1}\zeta_\mA, \label{eq: BL mass constituent material der: 4}\\
    \rho \partial_t c_\mA + \rho \bv \cdot \nabla c_\mA + {\rm div} (\bJ_\mA + \bj_\mA) &=~ \zeta_\mA,\label{eq: BL mass constituent material der: 5}\\
    \rho \partial_t c_\mA + \rho \bu \cdot \nabla c_\mA + {\rm div} (\mathbf{J}_\mA^u  + \bj_\mA) - c_\mA {\rm div}\left( \sum_\mB \bJ_\mB^u \right) &=~ \zeta_\mA. \label{eq: BL mass constituent material der: 6}
  \end{align}
\end{subequations}
Additionally, by invoking the relation \eqref{eq: relation c phi} we can deduce numerous alternative -- equivalent --  formulations; for example, by inserting \eqref{eq: relation c phi} into \eqref{eq: BL mass constituent material der: 3} we arrive at an uncommon formulation:
\begin{align}\label{eq: BL mass constituent alt} 
    \partial_t \left(\frac{c_\mA}{\rho_\mA}\left(\sum_\mB\frac{c_\mB}{\rho_\mB}\right)^{-1}\right) + {\rm div}\left(\frac{c_\mA}{\rho_\mA}\left(\sum_\mB\frac{c_\mB}{\rho_\mB}\right)^{-1} \bv\right) &\nn\\
    + {\rm div} \mathbf{h}_\mA + \rho_\mA^{-1}{\rm div} \bj_\mA & = \rho_\mA^{-1}\zeta_\mA.
\end{align}
Similarly, one can write the constituent momentum balance laws \eqref{eq: lin mom constituent j} as:
\begin{subequations}\label{eq: BL mom constituent material der}
  \begin{align}
    \partial_t (\trho_\mA\bv + \bJ_\mA) + {\rm div} \left( \trho_\mA\bv\otimes \bv +\bJ_\mA \otimes \bv + \bv\otimes \bJ_\mA  \right)&\nn\\
    - ~ {\rm div} \left(\mathbf{T}_\mA - \trho_\mA \bw_\mA\otimes \bw_\mA \right) -  \trho_\mA \mathbf{b}_\mA &=~ \boldsymbol{\pi}_\mA, \label{eq: BL mom constituent material der: 1}\\
    \partial_t (\trho_\mA\bu + \bJ_\mA^u) + {\rm div} \left( \trho_\mA\bu\otimes \bu +\bJ_\mA^u \otimes \bu + \bu\otimes \bJ_\mA^u  \right)&\nn\\
    + ~ {\rm div} \left(\trho_\mA \bomega_\mA\otimes \bomega_\mA - \trho_\mA \bw_\mA\otimes \bw_\mA  \right)&\nn\\
    -  {\rm div} \left(\mathbf{T}_\mA - \trho_\mA \bw_\mA\otimes \bw_\mA \right) -  \trho_\mA \mathbf{b}_\mA &=~ \boldsymbol{\pi}_\mA. \label{eq: BL mom constituent material der: 2}
  \end{align}
\end{subequations}
Finally, we introduce the constituent kinetic and gravitational energies, respectively, as:
\begin{subequations}
    \begin{align}
  \mathscr{K}_\mA =&~\trho_\mA \|\bv_\mA\|^2/2,\\
  \mathscr{G}_\mA =&~\trho_\mA b y,
\end{align}
\end{subequations}
where $\|\mathbf{v}_\mA\|=(\mathbf{v}_\mA \cdot \mathbf{v}_\mA)^{1/2}$ is the Euclidean norm of the velocity $\mathbf{v}_\mA$.

\subsection{Mixture balance laws}\label{sec: mixture BL}
The standard formulation of mixture balance laws is well-known and follows from summing the balance laws \eqref{eq: BL const} over all constituents. To establish the precise form, one can, for example, utilize the formulations \eqref{eq: BL mass constituent material der: 1} and \eqref{eq: BL mom constituent material der: 1} and invoke the identity \eqref{eq: rel gross motion zero: J} to obtain:
\begin{subequations}\label{eq: BL mix}
  \begin{align}
        \partial_t \rho + {\rm div}(\rho \bv) &=~ 0, \label{eq: local mass balance mix} \\
        \partial_t (\rho \bv) + {\rm div} \left( \rho \bv \otimes \bv \right) -  {\rm div} \mathbf{T} -  \rho \mathbf{b} &=~0,\label{eq: lin mom mix}\\
        \mathbf{T}-\mathbf{T}^T &=~0,\label{eq: ang mom mix}
  \end{align}
\end{subequations}
where the mixture stress and mixture body force are given by,  respectively:
\begin{subequations}
    \begin{align}
       \mathbf{T} =&~ \sum_\mA \mathbf{T}_\mA-\trho_\mA\bw_\mA\otimes\bw_\mA,\\
    \mathbf{b} =&~\frac{1}{\rho}\sum_\mA \trho_\mA\mathbf{b}_\mA,
    \end{align}
\end{subequations}
and where we have postulated the following balance conditions to hold as follows:
\begin{subequations}\label{eq: balance fluxes}
  \begin{align}
      \displaystyle\sum_\mA \gamma_\mA  =&~ 0,\label{eq: balance mass fluxes}\\
      \displaystyle\sum_\mA \boldsymbol{\pi}_\mA =&~ 0,\label{eq: balance momentum fluxes}\\
      \displaystyle\sum_\mA \mathbf{N}_\mA =&~ 0,
      \end{align}
\end{subequations}
and where we invoke \eqref{eq: balance mass fluxes} via:
\begin{subequations}\label{eq: BL conditions 2}
    \begin{align}
        \sum_\mA \zeta_\mA =&~ 0, \label{eq: BL conditions 2: zeta}\\
        \sum_\mA \mathbf{j}_\mA =&~ 0.\label{eq: BL conditions 2: bj} 
    \end{align}
\end{subequations}
This formulation is compatible with the first general principle: the motion of the mixture is derived from the motion of its individual constituents. In addition, the postulate \eqref{eq: balance fluxes} is essential to ensure general principle three. Even though the forms presented in \eqref{eq: BL mass constituent material der} and \eqref{eq: BL mom constituent material der} are equivalent, the summation of these laws over the constituents does not provide a suitable system of mixture balance laws for each of the formulations. Namely general principle three communicates that the resulting equations of the mixture are indistinguishable from that of a single body. Complying with this principle restricts the forms of the mass balance law to \eqref{eq: BL mass constituent material der: 1} and \eqref{eq: BL mass constituent material der: 2}, and requires the identification of suitable mixture variables. These variables are $\rho$, $\bv$, $\mathbf{T}$ and $\mathbf{b}$, as defined above. In this sense, the framework of continuum mixture theory serves as a guideline for defining mixture variables. However, one can work with other variables as well; and this is fully compatible with the framework.

We discuss other formulations that emerge from \eqref{eq: BL mass constituent material der} and \eqref{eq: BL mom constituent material der}. Summation of \eqref{eq: BL mass constituent material der: 2}-\eqref{eq: BL mass constituent material der: 6} over the constituents provides:
\begin{subequations}\label{eq: mix mass laws}
  \begin{align}
    \partial_t \rho + {\rm div}\left(\rho \left(\bu + \rho^{-1}\sum_\mA \mathbf{J}_\mA^u \right)\right) &=~ \sum_\mA \gamma_\mA = 0, \label{eq: mix mass laws: 2}\\
    {\rm div} \left(\bv + \sum_\mA \mathbf{h}_\mA \right) &=~ \sum_\mA \rho_\mA^{-1}\gamma_\mA, \label{eq: mix mass laws: 3} \\      
    {\rm div}\bu  &=~ \sum_\mA \rho_\mA^{-1}\gamma_\mA \label{eq: mix mass laws: 4}\\
    0 &=~ \sum_\mA \gamma_\mA, \label{eq: mix mass laws: 5}
  \end{align}
\end{subequations}
where \eqref{eq: mix mass laws: 2}-\eqref{eq: mix mass laws: 4} follow from \eqref{eq: BL mass constituent material der: 2}-\eqref{eq: BL mass constituent material der: 4}, respectively, and \eqref{eq: mix mass laws: 5} results from both \eqref{eq: BL mass constituent material der: 5} and \eqref{eq: BL mass constituent material der: 6}. We observe from \eqref{eq: mix mass laws: 2} that the term in the inner brackets in the second term represents the mixture velocity. Obviously, this matches the mass averaged velocity by invoking \cref{lem: mass av vs vol av}. Next, note that \eqref{eq: mix mass laws: 3} also follows from the summation over the constituents of \eqref{eq: BL mass constituent alt}. With the aid of \cref{lem: mass av vs vol av}, one can infer that \eqref{eq: mix mass laws: 3} and \eqref{eq: mix mass laws: 4} are identical. Furthermore, $\bv + \sum_\mA \bh_\mA = \bu$ is a divergence-free velocity whenever either (i) mass transfer is absent ($\gamma_\mA = 0$ for all $\mA$), or (ii) the constituent densities match ($\rho_\mA = \rho_\mB$ for all $\mA,\mB$). The equation \eqref{eq: mix mass laws: 5} complies with the balance condition \eqref{eq: balance mass fluxes}. Finally, the summation of \eqref{eq: BL mom constituent material der} yields:
\begin{align}
      \partial_t \left(\rho \left(\bu + \rho^{-1}\sum_\mA \bJ_\mA^u\right)\right) + {\rm div} \left( \rho \left(\bu\otimes \bu + \rho^{-1} \sum_\mA \bJ_\mA^u \otimes \bu + \rho^{-1}\bu\otimes \sum_\mA \bJ_\mA^u  \right)\right)\nn\\
     + ~ {\rm div} \left( \sum_\mA \trho_\mA \bomega_\mA\otimes \bomega_\mA - \trho_\mA \bw_\mA\otimes \bw_\mA  \right)-  {\rm div} \mathbf{T} -  \rho \mathbf{b} =~0.\hspace*{2cm}   
\end{align}
Invoking \cref{lem: mass av vs vol av} and \cref{lem: scaled peculiar velo} this may be written as:
\begin{align}
      \partial_t \left(\rho \left(\bu + \rho^{-1}\sum_\mA \bJ_\mA^u\right)\right) + {\rm div} \left( \rho \left(\bu\otimes \bu + \rho^{-1} \sum_\mA \bJ_\mA^u \otimes \bu + \rho^{-1}\bu\otimes \sum_\mA \bJ_\mA^u  \right)\right)\nn\\
    + ~ {\rm div} \left( \rho^{-1} \sum_\mA \bJ_\mA^u\otimes  \sum_\mA \bJ_\mA^u\right) -  {\rm div} \mathbf{T} -  \rho \mathbf{b} =~0.  \hspace*{2cm}  
\end{align}
One can infer equivalence with the mass-averaged momentum equation by noting the identities:
\begin{subequations}
\begin{align}
   {\rm div} \left( \rho \left(\bu\otimes \bu + \rho^{-1} \sum_\mA \bJ_\mA^u \otimes \bu + \rho^{-1}\bu\otimes \sum_\mA \bJ_\mA^u  \right)\right) 
   =\hspace*{2cm} \nn\\
    {\rm div} \left( \rho \bv \otimes \bv -\rho (\bv-\bu)\otimes(\bv-\bu)\right),\hspace*{2cm}  \\
   {\rm div} \left( \sum_\mA \trho_\mA \bomega_\mA\otimes \bomega_\mA - \trho_\mA \bw_\mA\otimes \bw_\mA  \right) = {\rm div} \left(\rho (\bv-\bu)\otimes(\bv-\bu)\right).\hspace*{2cm} 
\end{align}
\end{subequations}
In summary, an -- equivalent -- formulation of mixture balance laws \eqref{eq: BL mix} in terms of the volume-averaged velocity is:
\begin{subequations}\label{eq: BL mix vol}
  \begin{align}
        \partial_t \rho + {\rm div}\left(\rho \bu + \sum_\mA \bJ_\mA^u \right) &=~ 0, \label{eq: BL mix vol: mass} \\
              \partial_t \left(\rho\bu + \sum_\mA \bJ_\mA^u\right) + {\rm div} \left( \rho \bu\otimes \bu + \sum_\mA \bJ_\mA^u \otimes \bu  \right.& \nn\\
              + \left. \bu\otimes \sum_\mA \bJ_\mA^u  + \rho^{-1} \sum_\mA \bJ_\mA^u\otimes  \sum_\mA \bJ_\mA^u\right) -  {\rm div} \mathbf{T} -  \rho \mathbf{b} &=~0, \label{eq: BL mix vol: mom}\\
        \mathbf{T}-\mathbf{T}^T &=~0.\label{eq: BL mix vol: ang mom mix}
  \end{align}
\end{subequations}
The various forms presented in this section show that the set of balance laws, on both constituent level (\cref{sec: const BL}) and mixture level (\cref{sec: mixture BL}), is invariant to the set of fundamental variables.

We close this section with a remark on the kinetic and gravitational energies. According to the first metaphysical principle of mixture theory, the kinetic and gravitational energies of the mixture equal the summation of the constituent energies:
\begin{subequations}
  \begin{align}
  \mathscr{K} =&~ \displaystyle\sum_{\mA} \mathscr{K}_\mA,\label{eq: def sum K}\\
  \mathscr{G} =&~ \displaystyle\sum_{\mA} \mathscr{G}_\mA.
\end{align}
\end{subequations}
The kinetic energy of the mixture can be decomposed as:
\begin{subequations}\label{eq: relation kin energies}
    \begin{align}
      \mathscr{K} =&~ \bar{\mathscr{K}} + \displaystyle\sum_{\mA} \frac{1}{2} \trho_\mA \|\mathbf{w}_\mA\|^2,\\
       \bar{\mathscr{K}} =&~ \frac{1}{2} \rho \|\mathbf{v}\|^2,\label{eq: kin avg}
\end{align}
\end{subequations}
where $\bar{\mathscr{K}}$ represents the kinetic energy of the mixture variables, and where the other term is the kinetic energy of the constituents utilizing the peculiar velocity. The second terms may also be expressed in terms of volume-averaged quantities:
\begin{align}
  \displaystyle\sum_{\mA} \frac{1}{2} \trho_\mA \|\mathbf{w}_\mA\|^2 = 
  \displaystyle\sum_{\mA} \frac{1}{2} \trho_\mA \|\bomega_\mA - \rho^{-1} \sum_\mA \bJ_\mA^u \|^2. 
\end{align}

\section{Constitutive modeling}\label{sec: 2nd law}

This section details the development of constitutive models under the constraints of an energy-dissipative postulate. First, \cref{sec: const mod: subsec: def} outlines the fundamental assumptions and modeling choices. Next, \cref{sec: const mod: subsec: model restr} establishes the constitutive modeling restriction introduced in \cref{sec: const mod: subsec: def}, and \cref{sec: const mod: alt class} describes alternative modeling classes. Finally, in \cref{sec: const mod: subsec: select}, we select particular constitutive models that adhere to these established restrictions.

\subsection{Assumptions and modeling choices}\label{sec: const mod: subsec: def}

Rather than using the complete set of balance laws as given in \eqref{eq: local mass balance constituent j}, \eqref{eq: lin mom constituent j}, and \eqref{eq: ang mom constituent j}, we limit our focus to the simplified subset:
\begin{subequations}\label{eq: BL constitutive}
  \begin{align}
        \partial_t \phi_\mA + {\rm div}(\phi_\mA \bv) + \rho_\mA^{-1}{\rm div} \bH_\mA &=~ \rho_\mA^{-1}\zeta_\mA, \label{eq: BL constitutive: mass} \\
        \partial_t (\rho \bv) + {\rm div} \left( \rho\bv\otimes \bv \right) -  {\rm div} \mathbf{T} -  \rho \mathbf{b} &=~0,\label{eq: constitutive lin mom mix}\\
        \mathbf{T}-\mathbf{T}^T &=~0,\label{eq: constitutive ang mom mix}
  \end{align}
\end{subequations}
with $\bH_\mA := \bJ_\mA + \bj_\mA$, where \eqref{eq: BL constitutive: mass} holds for constituents $\mA=1,...,N$. At this point, the system is comprised of the unknown quantities: volume fractions $\phi_\mA$ ($\mA = 1,...,N$), where we recall the identity \eqref{eq: phi trho trho}, mass-averaged mixture velocity $\bv$, peculiar velocities $\bJ_\mA$ ($\mA = 1,...,N$), mass transfer terms $\zeta_\mA, \bj_\mA$ ($\mA = 1,...,N$), and mixture stress $\mathbf{T}$. In order to close the system we seek for constitutive models for $\bJ_\mA$, $\bj_\mA$, $\zeta_\mA$ and $\mathbf{T}$. Seeking for constitutive models for the peculiar velocities $\bJ_\mA$ could be perceived as a simplification procedure. Namely, substituting a constitutive model (in \cref{sec: const mod: subsec: select}), in general, violates the continuum mixture theory definitions \eqref{eq: def J and h}. We discard these definitions \eqref{eq: def J and h} in the following, but design models compatible with \cref{lem: mass av vs vol av} and \cref{lem: scaled peculiar velo} to ensure invariance with respect to the set of fundamental variables. Instead of working with $N$ velocities quantities $\bv_\mA$, the simplified system contains a single unknown velocity quantity $\bv$ and constitutive models for peculiar velocities $\bJ_\mA$. This is compatible with the structure of the system: the full system is composed of $N$ linear momentum (mixture) balance law whereas the simplified system contains a single linear momentum balance law. Additionally, we enforce the balance condition for the peculiar velocities \eqref{eq: rel gross motion zero: J} 
and the mass transfer terms \eqref{eq: balance mass fluxes} as follows:
\begin{subequations}\label{eq: BL conditions}
    \begin{align}
        \sum_\mA \bJ_\mA =&~ 0,\label{eq: BL conditions: bJ}\\
        \sum_\mA \bj_\mA =&~ 0,\label{eq: BL conditions: bj}\\
        \sum_\mA \zeta_\mA =&~ 0, \label{eq: BL conditions: gamma}
    \end{align}
\end{subequations}
where we recall the decomposition \eqref{eq: split gamma}.
The system \eqref{eq: BL constitutive} contains the unknown variables $\bv$ and $\phi_\mA$ ($\mA = 1,...,N$). We emphasize that the set $\left\{\phi_\mA\right\}_{\mA =1,...,N}$ is comprised of $N-1$ independent variables due to the summation condition \eqref{eq: sum phi}. As such, system \eqref{eq: BL constitutive} has a degenerate nature; it contains $N+1$ equations for $N$ variables (we preclude \eqref{eq: ang mom mix} in this count). \cref{sec: const mod: subsec: model restr} restores the balance by means of a Lagrange multiplier construction.

\begin{remark}[Classification]
  The above assumptions lead to a model that includes $N$ constituent mass balance laws along with a single momentum balance law. According to the classification by \cite{hutter2013continuum}, this configuration aligns best with a class-I model.
\end{remark}

We adopt the well-known Coleman-Noll procedure \citep{coleman1974thermodynamics} as a guiding principle to design constitutive models. For this purpose we postulate the energy-dissipation law:
\begin{align}\label{eq: energy dissipation}
    \dfrac{{\rm d}}{{\rm d}t} \mathscr{E} = \mathscr{W} - \mathscr{D},
\end{align}
satisfying $\mathscr{D}\geq 0$. The total energy is comprised of the Helmholtz free energy, the kinetic energy and the gravitational energy:
\begin{align}\label{eq: total energy}
  \mathscr{E} =  \displaystyle\int_{\mathcal{R}(t)}(\Psi + \bar{\mathscr{K}} + \mathscr{G})~{\rm d}v.
\end{align}
In this context, $\mathcal{R}=\mathcal{R}(t) \subset \Omega$ refers to a time-dependent control volume with volume element ${\rm d}v$ and a unit outward normal $\boldsymbol{\nu}$ that is transported by the velocity field $\bv$. Additionally, $\mathscr{W}$ represents a work rate term on the boundary $\partial \mathcal{R}(t)$ (with boundary element ${\rm d}a$), and $\mathscr{D}$ denotes the dissipation within the interior of $\mathcal{R}(t)$.

\begin{remark}[Energy-dissipation postulate]\label[remark]{remark: energy dissipation}
  As mentioned in \cite{eikelder2023unified,eikelder2023thermodynamically}, the energy-dissipation statement \eqref{eq: energy dissipation} can be perceived as approximation of the second law of thermodynamics for mixtures.
\end{remark}

We postulate that the free energy to pertain to the constitutive class:
\begin{align}\label{eq: class Psi}
  \Psi =&~ \hat{\Psi}\left(\left\{\phi_\mA\right\}_{\mA=1,...,N},\left\{\nabla \phi_\mA\right\}_{\mA=1,...,N}\right),
\end{align}
and introduce the chemical potential quantities ($\alpha = 1,...,N$):
\begin{align}\label{eq: chem pot}
    \hat{\mu}_\mA =&~ \dfrac{ \partial \hat{\Psi}}{\partial \phi_\mA} - {\rm div}\dfrac{\partial \hat{\Psi}}{\partial\nabla \phi_\mA}.
\end{align}
The volume fractions $\left\{\phi_\mA\right\}_{\mA=1,...,N}$ (and their gradients $\left\{\nabla \phi_\mA\right\}_{\mA=1,...,N}$) are not independent quantities due to the saturation constraint \eqref{eq: sum phi}. One may consider \eqref{eq: class Psi} and \eqref{eq: chem pot}
subject to the summation constraint \eqref{eq: sum phi}, or postpone enforcing it to the introduction of the Lagrange multiplier. In the latter case \eqref{eq: class Psi} and \eqref{eq: chem pot} are obviously well-defined, while we discuss some implications of the former case. 
As such, when considering the summation constraint \eqref{eq: sum phi}, the chemical potentials are individually arbitrary. For example, addition of the term $(1-\sum_\mA \phi_\mA) $ to $\hat{\Psi}$ does not alter it, but it modifies the chemical potentials $\mu_\mA$. 

\begin{remark}[Reduced free energy class]
Instead of utilizing the class \eqref{eq: class Psi}, one can also directly enforce the summation constraint \eqref{eq: sum phi} to arrive at a class with reduced dependency. In general, this breaks the symmetry of the approach, and therefore we do not adopt this alternative here. We discuss this option in \cref{appendix proofs}.
\end{remark}

\begin{remark}[Concentration-dependent free energy class]\label[remark]{remark: alternative free energy classes}
One can also work with a constituent class that depends on concentration quantities. We discuss this option in \cref{sec: const mod: alt class}. 
\end{remark}

\subsection{Modeling restriction}\label{sec: const mod: subsec: model restr}

Moving forward, we study in detail the restriction \eqref{eq: energy dissipation}. First, we analyze the evolution of the energy \eqref{eq: total energy}. Through the application of the Reynolds transport theorem to the free energy $\hat{\Psi}$, we have:
\begin{align}
      \dfrac{{\rm d}}{{\rm d}t}\displaystyle\int_{\mathcal{R}(t)} \hat{\Psi} ~{\rm d}v = \displaystyle\int_{\mathcal{R}(t)} \partial_t \hat{\Psi} ~{\rm d}v + \displaystyle\int_{\partial \mathcal{R}(t)} \hat{\Psi} \bv \cdot \bnu  ~{\rm d}a.
\end{align}

We notice that the summation constraint \eqref{eq: sum phi} does not alter the derivative of the free energy class \eqref{eq: class Psi}.
\begin{lemma}[Derivative of the free energy]\label[lemma]{lem: derivative}
The derivative of the free energy class \eqref{eq: class Psi} is given by:
\begin{align}
    {\rm d}\hat{\Psi} = \sum_{\mA} \dfrac{\partial \hat{\Psi}}{\partial \phi_\mA} {\rm d}\phi_\mA + \sum_{\mA} \dfrac{\partial \hat{\Psi}}{\partial \nabla\phi_\mA} {\rm d}(\nabla \phi_\mA),
\end{align}
where ${\rm d}$ is the derivative operator.
\end{lemma}

\begin{proof}
  See \cref{appendix: lem derivative}.
\end{proof}
Invoking \cref{lem: derivative} and the divergence theorem yields:
\begin{align}\label{eq: Psi derivation 1}
    \dfrac{{\rm d}}{{\rm d}t}\displaystyle\int_{\mathcal{R}(t)} \hat{\Psi} ~{\rm d}v  = \displaystyle\int_{\mathcal{R}(t)} &~ \hat{\Psi} ~{\rm div} \bv + \sum_{\mA}\dfrac{\partial \hat{\Psi}}{\partial \phi_\mA} \dot{\phi}_\mA+\sum_{\mA} \dfrac{\partial \hat{\Psi}}{\partial \nabla \phi_\mA}\cdot \left(\nabla \phi_\mA\right)\dot{} ~{\rm d}v.
\end{align}

Integrating by parts provides:
\begin{align}\label{eq: IP}
    \dfrac{{\rm d}}{{\rm d}t}\displaystyle\int_{\mathcal{R}(t)} \hat{\Psi} ~{\rm d}v  =&~ \displaystyle\int_{\mathcal{R}(t)} \hat{\Psi}~{\rm div} \bv + \sum_{\mA} \hat{\mu}_\mA \dot{\phi}_\mA  - \sum_{\mA} \nabla \phi_\mA \otimes \dfrac{\partial \hat{\Psi}}{\partial \nabla \phi_\mA}: \nabla \bv  ~{\rm d}v  \nn\\
    &~+ \displaystyle\int_{\partial \mathcal{R}(t)}\sum_{\mA} \dot{\phi}_\mA \dfrac{\partial \hat{\Psi}}{\partial \nabla \phi_\mA}\cdot \boldsymbol{\nu} ~{\rm d}a,
\end{align}
where we have substituted the identity
\begin{align}\label{eq: relation grad phi}
    (\nabla \uppsi)\dot{} = \nabla (\dot{\uppsi}) - (\nabla \uppsi)^T\nabla \bv
\end{align}
for $\uppsi = \phi_\mA$. An analysis of the free energy terms confirms their well-defined nature.
\begin{lemma}[Well-defined free energy terms]\label[lemma]{lem: well-defined}
The following free energy terms in \eqref{eq: IP} are well-defined:
\begin{align}
  \sum_{\mA} \hat{\mu}_\mA \dot{\phi}_\mA; \quad \sum_{\mA} \nabla \phi_\mA \otimes \dfrac{\partial \hat{\Psi}}{\partial \nabla \phi_\mA}; \quad \sum_{\mA} \dot{\phi}_\mA \dfrac{\partial \hat{\Psi}}{\partial \nabla \phi_\mA}.
\end{align}
\end{lemma}
\begin{proof}
See \cref{appendix: lem: well-defined}.
\end{proof}

Substituting the constituent mass balance laws \eqref{eq: BL constitutive: mass} provides:
\begin{align}
    \dfrac{{\rm d}}{{\rm d}t}\displaystyle\int_{\mathcal{R}(t)} \hat{\Psi} ~{\rm d}v  =&~ \displaystyle\int_{\mathcal{R}(t)} \hat{\Psi}~{\rm div} \bv + \sum_{\mA} \hat{\mu}_\mA \left(-\phi_\mA {\rm div}\bv - \rho_\mA^{-1}{\rm div} \bH_\mA + \rho_\mA^{-1}\zeta_\mA\right) \nn\\
    &- \sum_{\mA} \nabla \phi_\mA \otimes \dfrac{\partial \hat{\Psi}}{\partial \nabla \phi_\mA}: \nabla \bv  ~{\rm d}v  + \displaystyle\int_{\partial \mathcal{R}(t)}\sum_{\mA} \dot{\phi}_\mA \dfrac{\partial \hat{\Psi}}{\partial \nabla \phi_\mA}\cdot \boldsymbol{\nu} ~{\rm d}a,
\end{align}
where we recall $\bH_\mA = \bJ_\mA + \bj_\mA$. By again applying integration by parts one can infer that:
\begin{align}\label{eq: Psi 2}
    \dfrac{{\rm d}}{{\rm d}t}\displaystyle\int_{\mathcal{R}(t)} \sum_{\mA} \hat{\Psi} ~{\rm d}v = &~ \displaystyle\int_{\mathcal{R}(t)} \hat{\Psi}~{\rm div} \bv -\sum_{\mA} \hat{\mu}_\mA\phi_\mA {\rm div}\bv + \sum_{\mA} \nabla (\rho_\mA^{-1}\hat{\mu}_\mA) \cdot \bH_\mA \nn\\
    &~~~- \sum_{\mA} \nabla \phi_\mA \otimes \dfrac{\partial \hat{\Psi}}{\partial \nabla \phi_\mA}: \nabla \bv + \sum_{\mA}  \rho_\mA^{-1}\hat{\mu}_\mA \zeta_\mA~{\rm d}v\nn\\
    &~+ \displaystyle\int_{\partial \mathcal{R}(t)}\sum_{\mA} \left(\dot{\phi}_\mA \dfrac{\partial \hat{\Psi}}{\partial \nabla \phi_\mA}-\rho_\mA^{-1}\hat{\mu}_\mA \bH_\mA\right)\cdot \boldsymbol{\nu} ~{\rm d}a.
\end{align}
Next, the evolution of the kinetic and gravitational energies take the form (see \cite{eikelder2023unified} for details):
\begin{subequations}\label{eq: kin grav evo}
    \begin{align}
    \dfrac{{\rm d}}{{\rm d}t}\displaystyle\int_{\mathcal{R}(t)} \mathscr{K} ~{\rm d}v =&~ \displaystyle\int_{\mathcal{R}(t)} - \nabla \bv : \mathbf{T}+\rho\bv\cdot\mathbf{g}    ~{\rm d}v+ \displaystyle\int_{\partial \mathcal{R}(t)} \bv \cdot \mathbf{T} \bnu  ~{\rm d}a,\\
    \dfrac{{\rm d}}{{\rm d}t}\displaystyle\int_{\mathcal{R}(t)} \mathscr{G} ~{\rm d}v =&~- \displaystyle\int_{\mathcal{R}(t)} \rho \bv\cdot\mathbf{g}~{\rm d}v.
  \end{align}
  \end{subequations}
The superposition of \eqref{eq: Psi 2} and \eqref{eq: kin grav evo} provides the evolution of the total energy:
\begin{align}\label{eq: second law subst 1}
    \dfrac{{\rm d}}{{\rm d}t} \mathscr{E} = &~ \displaystyle\int_{\partial \mathcal{R}(t)}\left(\bv^T\mathbf{T}-\sum_{\mA} \left(\rho_\mA^{-1}\hat{\mu}_\mA\bH_\mA -\dot{\phi}_\mA \dfrac{\partial \hat{\Psi}}{\partial \nabla \phi_\mA}\right)\right)\cdot \boldsymbol{\nu} ~{\rm d}a \nn\\
    &~- \displaystyle\int_{\mathcal{R}(t)}   \left(\mathbf{T}  +\sum_{\mA} \nabla \phi_\mA \otimes \dfrac{\partial \hat{\Psi}}{\partial \nabla \phi_\mA} +\left(\sum_{\mA} \hat{\mu}_\mA\phi_\mA -  \hat{\Psi}\right)\mathbf{I}\right):\nabla \mathbf{v}\nn\\
    &~~~~~~~~~~~+ \sum_{\mA}\left(-\nabla (\rho_\mA^{-1}\hat{\mu}_\mA)\cdot \bH_\mA  - \rho_\mA^{-1}\hat{\mu}_\mA\zeta_\mA \right)~{\rm d}v.
\end{align}
As aforementioned in \cref{sec: const mod: subsec: def}, the system of balance laws \eqref{eq: BL constitutive} subject the balance conditions \eqref{eq: BL conditions} is degenerate. Namely, the terms $\nabla \mathbf{v}$, $\bH_\mA$ and $\zeta_\mA$ are connected via \eqref{eq: mix mass laws: 3}. This manifests itself in the energy dissipation statement \eqref{eq: second law subst 1}. The degeneracy needs to be eliminated in order to exploit the energy-dissipation condition as a guiding principle for constitutive modeling. To this purpose we enforce \eqref{eq: mix mass laws: 3} with the Lagrange multiplier construction:
\begin{align}\label{eq: LM 1}
    0=&~ \lambda \left( {\rm div}\bv + \displaystyle\sum_{\mA}  \rho_\mA^{-1} \nabla \cdot \bH_\mA  - \displaystyle\sum_{\mA}\rho_\mA^{-1}\zeta_\mA \right),
\end{align}
where $\lambda$ is the scalar Lagrange multiplier.

\begin{remark}[Lagrange multiplier constraint]
  Recalling \cref{lem: mass av vs vol av}, we observe that the Lagrange multiplier $\lambda$ enforces the constraint \eqref{eq: mix mass laws: 4}. As such, in absence of mass transfer ($\gamma_\mA=0$, $\mA=1,...,N$), it constrains ${\rm div}\bu=0$.
\end{remark}
Integrating \eqref{eq: LM 1} over $\mathcal{R}(t)$ and subtracting the result from \eqref{eq: second law subst 1} provides:
\begin{align}\label{eq: second law subst 2}
    \dfrac{{\rm d}}{{\rm d}t} \mathscr{E} = &~ \displaystyle\int_{\partial \mathcal{R}(t)}\left(\bv^T\mathbf{T}-\sum_{\mA} \left(g_\mA \bH_\mA -\dot{\phi}_\mA \dfrac{\partial \hat{\Psi}}{\partial \nabla \phi_\mA}\right)\right)\cdot \boldsymbol{\nu} ~{\rm d}a \nn\\
    &~- \displaystyle\int_{\mathcal{R}(t)}   \left(\mathbf{T} + \lambda\mathbf{I} +\sum_{\mA} \nabla \phi_\mA \otimes \dfrac{\partial \hat{\Psi}}{\partial \nabla \phi_\mA} +\left(\sum_{\mA} \hat{\mu}_\mA\phi_\mA -  \hat{\Psi}\right)\mathbf{I}\right):\nabla \mathbf{v}\nn\\
    &~~~~~~~~~~~~~~~~+ \sum_{\mA}\left(-\nabla g_\mA\cdot \bH_\mA - g_\mA \zeta_\mA \right)~{\rm d}v,
\end{align}
where we have utilized Gau{\ss} divergence theorem, and where we have defined the (generalized) chemical potential quantities:
\begin{subequations}
  \begin{align}
    g_\mA :=&~ \rho_\mA^{-1}\hat{\mu}_{\mA,\lambda},\\
    \hat{\mu}_{\mA,\lambda} :=&~ \hat{\mu}_\mA +\lambda. \label{eq: def mu alpha lambda}
  \end{align}
\end{subequations}
We identify the rate of work and the dissipation respectively as:
\begin{subequations}\label{eq: W, D}
\begin{align}
    \mathscr{W} =&~ \displaystyle\int_{\partial \mathcal{R}(t)}\left(\bv^T\mathbf{T}-\sum_{\mA} \left(g_\mA\bH_\mA -\dot{\phi}_\mA \dfrac{\partial \hat{\Psi}}{\partial \nabla \phi_\mA}\right)\right)\cdot \boldsymbol{\nu} ~{\rm d}a,\\
    \mathscr{D} =&~\displaystyle\int_{\mathcal{R}(t)}   \left(\mathbf{T} + \lambda\mathbf{I} +\sum_{\mA} \nabla \phi_\mA \otimes \dfrac{\partial \hat{\Psi}}{\partial \nabla \phi_\mA} +\left(\sum_{\mA} \hat{\mu}_\mA\phi_\mA -  \hat{\Psi}\right)\mathbf{I}\right):\nabla \mathbf{v}\nn\\
    &~~~~~~~~~~~~~+ \sum_{\mA}\left(-\nabla g_\mA\cdot \bH_\mA -g_\mA\zeta_\mA \right)~{\rm d}v.\label{eq: def diffusion}
\end{align}
\end{subequations}
Given the arbitrary nature of the control volume $\mathcal{R}=\mathcal{R}(t)$, the fulfillment of the energy dissipation law is contingent upon satisfying the local inequality:
\begin{align}\label{eq: second law 4}
     \left(\mathbf{T} + \lambda\mathbf{I} +\sum_{\mA} \nabla \phi_\mA \otimes \dfrac{\partial \hat{\Psi}}{\partial \nabla \phi_\mA} +\left(\sum_{\mA} \hat{\mu}_\mA\phi_\mA -  \hat{\Psi}\right)\mathbf{I}\right):\nabla \mathbf{v}&\nn\\
    - \sum_{\mA}\nabla g_\mA\cdot \bH_\mA -\sum_{\mA} g_\mA \zeta_\mA &\geq 0.
\end{align}

\begin{remark}[Compatibility with continuum mixture theory]
  This section has demonstrated that the energy dissipation postulate \eqref{eq: energy dissipation} is fulfilled when the local inequality \eqref{eq: second law 4} is satisfied. As mentioned in \cref{remark: energy dissipation}, the energy dissipation postulate is an approximation of the second law of mixture theory. However, we emphasize that the presented derivations are fully compatible with continuum mixture theory. 
\end{remark}

We finalize this section with a remark on the connection between the chemical potentials and the Lagrange multiplier.
Recalling the saturation constraint \eqref{eq: sum phi}, we recognize:
\begin{align}\label{eq: lambda mu}
    \lambda + \sum_{\mA}\hat{\mu}_\mA \phi_\mA = \sum_\mA \hat{\mu}_{\mA,\lambda} \phi_\mA.
\end{align}
This observation reveals that chemical potentials in \eqref{eq: second law 4} solely occur in the form $\hat{\mu}_{\mA,\lambda}$. In other words, the chemical potentials $\hat{\mu}_\mA$ are tightly connected with the Lagrange multiplier $\lambda$. This is consistent with the examination that the addition of $\sum_\mA \phi_\mA - 1$ does not alter the free energy. Indeed, we have:
\begin{align}\label{eq: class Psi Lambda}
  \Psi =&~ \hat{\Psi}\left(\left\{\phi_\mA\right\}_{\mA=1,...,N},\left\{\nabla \phi_\mA\right\}_{\mA=1,...,N}\right) + \lambda \left(\sum_\mA \phi_\mA -1 \right),
\end{align}
and the associated chemical potential quantities ($\alpha = 1,...,N$) naturally include the Lagrange multiplier $\lambda$:
\begin{align}\label{eq: chem pot 1}
    \hat{\mu}_{\mA,\lambda} =  \dfrac{ \partial \Psi}{\partial \phi_\mA} - {\rm div}\dfrac{\partial \Psi}{\partial\nabla \phi_\mA},
\end{align}
where we recall \eqref{eq: def mu alpha lambda}.

\subsection{Alternative free energy classes}\label{sec: const mod: alt class}

As mentioned in \cref{remark: alternative free energy classes}, as an alternative for \eqref{eq: class Psi}, we explore the approach of working with a class that depends on concentration. This exploration is motivated by its occurrence in the literature on two-phase models (e.g. \cite{lowengrub1998quasi}). We consider the following  constitutive class:
\begin{align}\label{eq: class Psi c}
  \Psi =&~ \check{\Psi}\left(\left\{c_\mA\right\},\left\{\nabla c_\mA\right\}\right),
\end{align}
subject to the summation constraint \eqref{eq: sum c}.  Alongside with free energy class \eqref{eq: class Psi c} we introduce the chemical potential quantities ($\alpha = 1,...,N$):
\begin{align}
    \check{\mu}_\mA =&~ \dfrac{ \partial \hat{\Psi}}{\partial c_\mA} - {\rm div}\dfrac{\partial \hat{\Psi}}{\partial\nabla c_\mA}.
\end{align}
In \cref{appendix: alt const mod} we provide the derivation of the modeling restriction that emerges from the constitutive class \eqref{eq: class Psi c}. The modeling restriction takes the form:
\begin{align}\label{eq: restriction alt}
     \left(\mathbf{T} + \check{\lambda}\mathbf{I} +\sum_{\mA} \nabla c_\mA \otimes \dfrac{\partial \check{\Psi}}{\partial \nabla c_\mA}  -  \check{\Psi}\mathbf{I}\right):\nabla \mathbf{v}&\nn\\
    - \sum_{\mA}\nabla \left( \rho^{-1}\check{\mu}_\mA + \rho_\mA^{-1}\check{\lambda} \right)\cdot \bH_\mA -\sum_{\mA} \left( \rho^{-1}\check{\mu}_\mA + \rho_\mA^{-1}\check{\lambda} \right) \zeta_\mA &\geq 0,
\end{align}
where $\check{\lambda}$ is the Lagrange multiplier associated with the constraint \eqref{eq: mix mass laws: 3}. Noting that the volume fractions and concentrations are connected via \eqref{eq: relation c phi}:
\begin{subequations}\label{eq: relation c phi 2}
    \begin{align}
        \phi_\mA =&~ \phi_\mA\left(\left\{c_\mB\right\}\right),\\
        c_\mB =&~ c_\mB\left(\left\{\phi_\mA\right\}\right),
    \end{align}
\end{subequations}
the identification
\begin{align}\label{eq: identify free energy}
    \hat{\Psi}\left(\left\{\phi_\mA\right\}\right) =\check{\Psi}\left(\left\{c_\mB\right\}\right)
\end{align}
reveals that the free energy classes coincide. Given that the initial modeling restriction is the same for both classes we conclude that the resulting modeling restrictions must coincide as well. In other words, \textit{the modeling restriction is independent of the choice of order parameters}. 
\begin{theorem}[Equivalence of modeling restrictions]\label{thm: mod restrictions}
  The modeling restrictions \eqref{eq: second law 4} and \eqref{eq: restriction alt} are equivalent.
\end{theorem}
An alternative path to show equivalence of the modeling restrictions, one could apply the variable transformation \eqref{eq: relation c phi 2} defined in \eqref{eq: relation c phi} to show that \eqref{eq: restriction alt} coincides with \eqref{eq: second law 4}. 
We discuss this approach in \cref{appendix: sec: Equivalence of modeling restrictions}.

Guided by \cref{thm: mod restrictions}, we proceed with the formulation of the modeling restriction presented in \eqref{eq: second law 4}.

\subsection{Selection of constitutive models}\label{sec: const mod: subsec: select}

By means of the Colemann-Noll concept, we utilize \eqref{eq: second law 4} as a guiding principle to design constitutive models. Inspired by the specific form of the constraint \eqref{eq: second law 4}, we restrict ourselves to mixture stress tensors $\mathbf{T}$, constituent peculiar velocities $\bJ_\mA$, and constituent mass transfer terms $\bj_\mA, \zeta_\mA$ that belong to the constitutive classes:
\begin{subequations}\label{eq: class T J gamma}
\begin{align}
    \mathbf{T} =&~ \hat{\mathbf{T}}\left(\nabla \bv, \left\{\phi_\mA\right\}, \left\{\nabla \phi_\mA\right\}, \left\{g_{\mA}\right\}, \left\{\nabla g_{\mA}\right\}\right),\label{eq: class T}\\
    \bJ_\mA =&~ \hat{\bJ}_\mA\left(\left\{\phi_\mA\right\}, \left\{\nabla g_\mA\right\}\right),\label{eq: class J}\\
    \bj_\mA =&~ \hat{\bj}_\mA\left(\left\{\phi_\mA\right\}, \left\{\nabla g_\mA\right\}\right),\label{eq: class j}\\
    \zeta_\mA =&~ \hat{\zeta}_\mA\left(\left\{\phi_\mA\right\}, \left\{g_\mA\right\}\right),\label{eq: class zeta}
\end{align}
\end{subequations}
and define $\hat{\bH}_\mA = \hat{\bJ}_\mA + \hat{\bj}_\mA$. Generally speaking, the introduction of the class \eqref{eq: class J} deviates from continuum mixture theory. Arguably, a natural approximation is simply taking $\hat{\bJ}_\mA = 0$, which, for instance, models the situation of matching velocities $\bv_\mA=\bv_\mB$. We return to this case in \cref{sec: Matching velocities}.

We do not seek the most complete constitutive theory, rather our goal is to find a set of practical constitutive models compatible with \eqref{eq: second law 4}. To this end, we aim to identify constitutive models \eqref{eq: class T J gamma} so that all three terms in \eqref{eq: second law 4} are positive, which occurs when:
\begin{subequations}\label{eq: restrictions}
    \begin{align}
     \left(\mathbf{T} +\sum_{\mA} \nabla \phi_\mA \otimes \dfrac{\partial \hat{\Psi}}{\partial \nabla \phi_\mA} +\left(\sum_{\mA} \hat{\mu}_{\mA,\lambda}\phi_\mA -  \hat{\Psi}\right)\mathbf{I}\right):\nabla \mathbf{v}&\geq 0, \label{eq: restrictions T}\\
     -\sum_{\mA}\nabla g_\mA\cdot \bJ_\mA &\geq 0, \label{eq: restrictions J}\\
     -\sum_{\mA}\nabla g_\mA\cdot \bj_\mA &\geq 0, \label{eq: restrictions j}\\
-\sum_{\mA} g_\mA\zeta_\mA &\geq 0. \label{eq: restrictions gamma}
\end{align}
\end{subequations}

\begin{remark}[Onsager reciprocal relations]\label[remark]{remark: onsager}
  As mentioned above, our objective is to find a set of practical constitutive models. A more complete theory follows from working with the original constraint \eqref{eq: second law 4}, and extending the dependency of the classes \eqref{eq: class T J gamma}. In particular, the classes may be interconnected. The well-known Onsager reciprocal relations take a central place in this framework. We refer to \cite{onsager1931reciprocalI,onsager1931reciprocalII}.
\end{remark}

In the following we provide constitutive models for the mixture stress tensor, constituent peculiar velocities, and constituent mass transfer, respectively.\\

\noindent \textit{Mixture stress tensor}. We select the following constitutive model for the stress tensor:
\begin{align}\label{eq: stress tensor choice}
    \hat{\mathbf{T}} = -\sum_{\mA} \nabla \phi_\mA \otimes \dfrac{\partial \hat{\Psi}}{\partial \nabla \phi_\mA}-\left(\sum_{\mA} \hat{\mu}_{\mA,\lambda}\phi_\mA -  \hat{\Psi}\right)\mathbf{I} + \nu (2\nabla^s \bv+\bar{\lambda}({\rm div}\bv) \mathbf{I}),
\end{align}
subject to the symmetry condition:
\begin{align}\label{eq: T symm}
   \nabla \phi_\mA \otimes \dfrac{\partial \hat{\Psi}}{\partial \nabla \phi_\mA} =  \dfrac{\partial \hat{\Psi}}{\partial \nabla \phi_\mA} \otimes \nabla \phi_\mA,
  \end{align}
where the scalar field $\nu\geq 0$ is the mixture dynamic viscosity, $\bar{\lambda} \geq -2/d$ is a scalar, and $d$ is the number of dimensions. {\color{red}Possible choices for the mixture viscosity include $\nu = \sum_\mA \nu_\mA \phi_\mA$ and $\nu = \sum_\mA \nu_\mA c_\mA$, where $\nu_\mA$ are constituent viscosities.} The condition \eqref{eq: T symm} ensures compatibility with the angular momentum constraint \eqref{eq: ang mom mix}.

\begin{lemma}[Compatibility mixture stress tensor]\label[lemma]{lem: compatibility stress tensor}
The mixture stress tensor \eqref{eq: stress tensor choice} adheres to the constraint \eqref{eq: restrictions T}.
\end{lemma}
\begin{proof}
  An elementary calculation gives:
  \begin{align}\label{eq: second law viscosity}
  \left(\hat{\mathbf{T}}  +\sum_{\mA} \nabla \phi_\mA \otimes \dfrac{\partial \hat{\Psi}}{\partial \nabla \phi_\mA} +\left(\sum_{\mA} \hat{\mu}_{\mA,\lambda}\phi_\mA -  \hat{\Psi}\right)\mathbf{I}\right):\nabla \mathbf{v} &= \nn \\   2 \nu \left( \nabla^s \bv - \frac{1}{d} ({\rm div} \mathbf{v}) \mathbf{I}\right):\left(\nabla^s \bv - \frac{1}{d} ({\rm div} \mathbf{v}) \mathbf{I}\right) + \nu \left(\bar{\lambda} + \frac{2}{d}\right)\left({\rm div} \mathbf{v}\right)^2 &\geq 0.
  \end{align}
\end{proof}

\noindent \textit{Constituent peculiar velocities}. We choose the peculiar velocities of the form:
\begin{align}\label{eq: model J}
    \hat{\bJ}_\mA = -\sum_\mB \mathbf{M}_{\mA\mB} \nabla g_\mB,
\end{align}
with mobility tensor $\mathbf{M}_{\mA\mB}=\mathbf{M}_{\mB\mA}$. The mobility tensor is positive definite ($\by_\mA^T \mathbf{M}_{\mA\mB}\by_\mB \geq 0$ for all $\by_\mA \in \mathbb{R}^d, \mA = 1,...,N$), has the same dependencies as \eqref{eq: restrictions J}, is compatible with $\sum_\mA \mathbf{M}_{\mA\mB} = \sum_\mA \mathbf{M}_{\mB\mA} = 0$ for all $\mB=1,..,N$, and vanishes in the single fluid region $\mathbf{M}_{\mA\mB}|_{\phi_\mC=1}=0, \mC =1,...,N$ (thus is degenerate). We note that the symmetry requirement follows from the Onsager reciprocal relations, the positive definiteness from \eqref{eq: restrictions J}, and the zero sum of rows and columns from \eqref{eq: BL conditions: bJ}. {\color{red}A possible choice for the mobility tensor is $\mathbf{M}_{\mA\mB}=-\mathbf{M}_0 \trho_\mA \trho_\mB$ for $\mA \neq \mB$, and $\mathbf{M}_{\mA\mA} = \mathbf{M}_{0}\trho_\mA \sum_{\mC\neq \mA} \trho_\mC$ for some $\mathbf{M}_0$ that is not dependent on the constituent number.}


\begin{remark}[Lagrange multiplier in constituent peculiar velocities]
 In most incompressible $N$-phase models, the Lagrange multiplier $\lambda$ does not explicitly appear in the constituent peculiar velocities $\hat{\bJ}_\mA$, whereas in the proposed framework, it appears as a component of $g_\mA$. Notably, when all constituent densities are identical ($\rho_\mA = \rho$ for $\mA=1,...,N$), the Lagrange multiplier vanishes, yielding the relation: $g_\mA - g_\mB = \rho^{-1}(\hat{\mu}_\mA-\hat{\mu}_\mB)$.
  \end{remark}

\begin{lemma}[Compatibility constituent peculiar velocities]\label[lemma]{lem: compatibility peculiar velocity}
The choice \eqref{eq: model J} aligns with both the balance \eqref{eq: rel gross motion zero}, and the restriction \eqref{eq: restrictions J}.
\end{lemma}

\noindent \textit{Constituent diffusive flux}. 
Analogously to the constituent peculiar velocities, we select:
\begin{align}\label{eq: model j}
    \hat{\bj}_\mA = -\sum_\mB \mathbf{K}_{\mA\mB} \nabla g_\mB,
\end{align}
for some positive definite constitutive tensor $\mathbf{K}_{\mA\mB}=\mathbf{K}_{\mB\mA}$ compatible with $\sum_\mA \mathbf{K}_{\mA\mB}=\sum_\mA \mathbf{K}_{\mB\mA}=0$, with the same dependencies as \eqref{eq: class j}, and vanishes in the single fluid region $\mathbf{K}_{\mA\mB}|_{\phi_\mC=1}=0, \mC =1,...,N$. {\color{red}Similarly as for the peculiar velocity, a possible choice for the mobility tensor is $\mathbf{K}_{\mA\mB}=-\mathbf{K}_0 \trho_\mA \trho_\mB$ for $\mA \neq \mB$, and $\mathbf{K}_{\mA\mA} = \mathbf{K}_{0}\trho_\mA \sum_{\mC\neq \mA} \trho_\mC$ for some $\mathbf{K}_0$ that is not dependent on the constituent number.}
\begin{lemma}[Compatibility constituent diffusive fluxes]\label[lemma]{lem: compatibility diffusive flux}
The choice \eqref{eq: model j} aligns with both the balance \eqref{eq: BL conditions 2: bj}, and the restriction \eqref{eq: restrictions j}.
\end{lemma}

\noindent \textit{Constituent mass transfer}. We select the constituent mass transfer terms analogously to the constituent peculiar velocities:
\begin{align}\label{eq: const model mass flux}
  \hat{\zeta}_\mA = -\sum_\mB m_{\mA\mB} g_\mB,
\end{align}
where the positive definite scalar mobility $m_{\mA\mB}=m_{\mB\mA}$ has the same dependencies as \eqref{eq: class zeta}, is compatible with $\sum_\mA m_{\mA\mB} = \sum_\mA m_{\mB\mA} = 0$, vanishes in the single fluid region $m_{\mA\mB}|_{\phi_\mC=1}=0, \mC =1,...,N$. 

\begin{lemma}[Compatibility mass transfer]\label[lemma]{lem: compatibility gamma}
The choice \eqref{eq: const model mass flux} is compatible with the balance of mass supply \eqref{eq: balance momentum fluxes}, and the constraint \eqref{eq: restrictions gamma}.
\end{lemma}
\begin{remark}[Related constitutive models]\label[remark]{rmk: Related constitutive models}
In the case:
  \begin{align}
      \mathbf{M}_{\mA\mB} = \begin{cases}
          -\hat{\mathbf{M}}_{\mA\mB} & \text{ if }\mA \neq \mB,\\
          \sum_{\mC\neq \mA} \hat{\mathbf{M}}_{\mA\mC} & \text{ if }\mA = \mB,
      \end{cases}
  \end{align}
  for some symmetric $\hat{\mathbf{M}}_{\mA\mB}$, we find:
  \begin{align}\label{eq: alternative form J}
      \hat{\bJ}_\mA =&~ - \sum_{\mB\neq \mA} \mathbf{M}_{\mA\mB}\nabla \hat{\mu}_\mB - \mathbf{M}_{\mA\mA}\nabla \hat{\mu}_\mA\nn\\
      =&~  \sum_{\mB\neq \mA} \hat{\mathbf{M}}_{\mA\mB}\nabla \hat{\mu}_\mB - \sum_{\mC\neq\mA}\hat{\mathbf{M}}_{\mA\mC}\nabla \hat{\mu}_\mA\nn\\
      =&~ -\sum_{\mB} \hat{\mathbf{M}}_{\mA\mB}\nabla (\hat{\mu}_\mA-\hat{\mu}_\mB).
  \end{align}
  This model matches (for the isotropic case $\mathbf{M}_{\mA\mB}=M_{\mA\mB}\mathbf{I}$) that of \cite{eikelder2023thermodynamically}. It also closely resembles the form adopted in \cite{li2014class}. Both closure models involve the Lagrange multiplier $\lambda$; a difference lies in fact that the model proposed by \cite{li2014class} depends on the numbering of the constituents. Finally, we note that forms similar to \eqref{eq: alternative form J} may be adopted for the diffusive fluxes and the mass transfer terms.
\end{remark}

This finalizes the construction of constitutive models compatible with the imposed energy-dissipative postulate. Substitution of the models \eqref{eq: stress tensor choice}, \eqref{eq: model J}, and \eqref{eq: const model mass flux} yields the class of incompressible $N$-phase models: 
\begin{subequations}\label{eq: model full 0}
  \begin{align}
   \partial_t (\rho \bv) + {\rm div} \left( \rho \bv\otimes \bv \right) + \nabla \lambda & \nn\\
   + {\rm div} \left(\left(\sum_\mA \hat{\mu}_\mA\phi_\mA -\hat{\Psi}\right)\mathbf{I} + \sum_\mA \nabla \phi_\mA \otimes \dfrac{\partial \hat{\Psi}}{\partial \nabla \phi_\mA} \right)& \nn\\
    - {\rm div} \left(   \nu (2\nabla^s \bv+\bar{\lambda}({\rm div}\bv) \mathbf{I}) \right)-\rho\mathbf{g} &=~ 0, \label{eq: model full 0: mom}\\
  \partial_t \phi_\mA  + {\rm div}(\phi_\mA  \bv) + \rho_\mA^{-1}{\rm div}\hat{\mathbf{J}}_\mA + \rho_\mA^{-1}{\rm div}\hat{\mathbf{j}}_\mA - \hat{\zeta}_\mA &=~0,\label{eq: model full 0: mass}\\
  \hat{\mu}_\mA - \dfrac{\partial \hat{\Psi}}{\partial \phi_\mA} + {\rm div} \left(  \dfrac{\partial \hat{\Psi}}{\partial \nabla \phi_\mA} \right)&=~0,\\
  \hat{\bJ}_\mA + \sum_\mB \mathbf{M}_{\mA\mB} \nabla g_\mB&~=0,\\
  \hat{\bj}_\mA + \sum_\mB \mathbf{K}_{\mA\mB}\nabla g_\mB &~=0,\\
  \hat{\zeta}_\mA + \sum_\mB m_{\mA\mB} g_\mB &~=0.
  \end{align}
\end{subequations}
Formulation \eqref{eq: model full 0} constitutes a class of models in the sense that particular closure relations ($\mathbf{M}_{\mA\mB}$, $\mathbf{K}_{\mA\mB}$ and $m_{\mA\mB}$) need to be specified. Given these relations, for each specification of the free energy $\hat{\Psi}$, \eqref{eq: model full 0} is a well-defined closed model; this model is invariant to the renumbering of the constituents, invariant to the set of independent variables, and reduction-consistent. Additionally, it exhibits energy-dissipation, which we explicitly state in the following theorem.
\begin{theorem}[Compatibility energy dissipation]\label{eq: compatibility energy dissipation}
The model \eqref{eq: model full 0} is compatible with the energy dissipation condition \eqref{eq: energy dissipation}.
\end{theorem}
\begin{proof}
  This follows from \cref{lem: compatibility stress tensor}, \cref{lem: compatibility peculiar velocity}, and \cref{lem: compatibility gamma}. In particular, the dissipation takes the form:
\begin{align}\label{eq: second law red 7}
    \mathscr{D} = &~ \displaystyle\int_{\mathcal{R}(t)} 2 \nu \left( \nabla^s \bv - \frac{1}{d} ({\rm div} \mathbf{v}) \mathbf{I}\right):\left(\nabla^s \bv - \frac{1}{d} ({\rm div} \mathbf{v}) \mathbf{I}\right)+ \nu \left(\lambda + \frac{2}{d}\right)\left({\rm div} \mathbf{v}\right)^2 \nn\\
    &~\quad\quad+ \displaystyle\sum_{\mA,\mB} \left(\nabla g_\mA\right)^T \mathbf{B}_{\mA\mB} \nabla g_\mB+ \sum_{\mA,\mB} m_{\mA\mB}g_\mA g_\mB~{\rm d}v \geq 0,
\end{align}
with $\mathbf{B}_{\mA\mB}=\mathbf{M}_{\mA\mB}+\mathbf{K}_{\mA\mB}$.
\end{proof}

\section{Model characteristics}\label{sec: Properties}

In this section, we explore the characteristics of the modeling framework outlined in \cref{sec: 2nd law}. To this end, we discuss alternative -- equivalent -- formulations in \cref{sec: Properties: alt}. We present the case of matching velocities in \cref{sec: Matching velocities}. Subsequently, \cref{sec: Properties: equi} details the equilibrium characteristics.

\subsection{Alternative formulations}\label{sec: Properties: alt}

As discussed in \cref{sec: mix theory,sec: 2nd law}, the unified modeling framework outlined in these sections is invariant to the choice of variables. However, it is worthwhile to discuss some of the formulations that are associated with particular variables.

First, we note that one can identify a pressure quantity in the model as:
\begin{align}\label{eq: def pressure}
  p: = \sum_\mA \hat{\mu}_\mA\phi_\mA -\hat{\Psi}.
\end{align}
With this choice, the model takes the more compact form:
\begin{subequations}\label{eq: model full 2}
  \begin{align}
   \partial_t (\rho \bv) + {\rm div} \left( \rho \bv\otimes \bv \right) + \nabla (\lambda + p) + {\rm div} \left( \sum_\mB \nabla \phi_\mB \otimes \dfrac{\partial \hat{\Psi}}{\partial \nabla \phi_\mB} \right)& \nn\\
    - {\rm div} \left(   \nu (2\nabla^s \bv+\bar{\lambda}({\rm div}\bv) \mathbf{I}) \right)-\rho\mathbf{g} &=~ 0, \label{eq: model full 2: mom}\\
  \partial_t \phi_\mA  + {\rm div}(\phi_\mA  \bv) + \rho_\mA^{-1}{\rm div}\hat{\bH}_\mA - \hat{\zeta}_\mA &=~0.\label{eq: model full 2: mass}
  \end{align}
\end{subequations}

In accordance with the first metaphysical principle of continuum mixture theory, the mixture free energy is comprised of constituent free energies:
\begin{align}\label{eq: split free energy}
    \hat{\Psi} = \sum_\mA \hat{\Psi}_\mA,
\end{align}
where $\hat{\Psi}_\mA$ are the volume-measure constituent free energies. Utilizing \eqref{eq: split free energy} we observe that the pressure satisfies Dalton's law:
\begin{subequations}
    \begin{align}
        p     = &~ \sum_\mA p_\mA, \\
        p_\mA = &~ \hat{\mu}_\mA \phi_\mA - \hat{\Psi}_\mA,
    \end{align}
\end{subequations}
where $p_\mA$ is the partial pressure of constituent $\mA$. Thus, the split \eqref{eq: split free energy} reveals that the system may be written as:
\begin{subequations}\label{eq: model pressure}
  \begin{align}
   \sum_\mB \left( \partial_t (\trho_\mB \bv) + {\rm div} \left( \trho_\mB \bv\otimes \bv \right) + \phi_\mB \nabla \lambda + \nabla p_\mB -\trho_\mB \mathbf{g} 
   \right. &\nn\\
   \left.+~ {\rm div} \left( \sum_\mC \nabla \phi_\mC \otimes \dfrac{\partial \hat{\Psi}_\mB}{\partial \nabla \phi_\mC} \right)\right) 
   - {\rm div} \left(   \nu (2\nabla^s \bv+\lambda({\rm div}\bv) \mathbf{I}) \right) &=~ 0, \label{eq: model  simplified: mom}\\
  \partial_t \phi_\mA  + {\rm div}(\phi_\mA  \bv) + \rho_\mA^{-1}{\rm div}\hat{\bH}_\mA - \hat{\zeta}_\mA &=~0.\label{eq: eq: model simplified: mass}
  \end{align}
\end{subequations}

An alternative compact form is obtained with the aid of the following lemma.
\begin{lemma}[Identity free energy]\label[lemma]{lem: identity Korteweg stresses}
The free energy contributions collapse into:
  \begin{align}
       {\rm div} \left(\left(\sum_\mA \hat{\mu}_\mA\phi_\mA -\hat{\Psi}\right)\mathbf{I} + \sum_\mA \nabla \phi_\mA \otimes \dfrac{\partial \hat{\Psi}}{\partial \nabla \phi_\mA} \right) =  \sum_\mA\phi_\mA\nabla\hat{\mu}_\mA.
  \end{align}
\end{lemma}
\begin{proof}
See \cref{lem: free energy identity}.
\end{proof}
Invoking \cref{lem: identity Korteweg stresses}, model \eqref{eq: model full 0} takes a more compact form:
\begin{subequations}\label{eq: model full}
  \begin{align}
   \partial_t (\rho \bv) + {\rm div} \left( \rho \bv\otimes \bv \right) + \nabla \lambda + \sum_\mA\phi_\mA\nabla\hat{\mu}_\mA &\nn\\ - {\rm div} \left(   \nu (2\nabla^s \bv+\bar{\lambda}({\rm div}\bv) \mathbf{I}) \right)-\rho\mathbf{g} &=~ 0, \label{eq: model full: mom}\\
  \partial_t \phi_\mA  + {\rm div}(\phi_\mA  \bv) + \rho_\mA^{-1}{\rm div}\hat{\bH}_\mA - \hat{\zeta}_\mA &=~0.\label{eq: model full: mass}
  \end{align}
\end{subequations}
Considering the third and fourth term in the momentum equation in isolation, these can be written as:
  \begin{align}\label{eq: identity lambda mu}
    \nabla \lambda + \sum_\mA\phi_\mA\nabla\hat{\mu}_\mA = 
    \sum_\mA \phi_\mA \nabla (\lambda +  \hat{\mu}_\mA) = \sum_\mA \trho_\mA \nabla g_\mA. 
  \end{align}
Similarly, in the mass balance \eqref{eq: model full: mass}, we observe that the chemical potentials $\hat{\mu}_\mA$ and the Lagrange multiplier $\lambda$ appear solely as a sum via $g_\mA$.

Additionally, we note that the model can alternatively be written in a form that more closely links to existing phase-field models:
\begin{subequations}\label{eq: model alt}
  \begin{align}
   \partial_t (\rho \bv) + {\rm div} \left( \rho \bv\otimes \bv \right) + \nabla \lambda +  \sum_{\mB}  \phi_\mB\nabla\hat{\mu}_\mB 
  &\nn\\  - {\rm div} \left(   \nu (2\nabla^s \bv+\bar{\lambda}({\rm div}\bv) \mathbf{I}) \right)-\rho\mathbf{g} &=~ 0, \label{eq: model alt: mom}\\
  \partial_t \rho + {\rm div}(\rho \bv) &=~0,\label{eq: eq: model alt: mass full}\\
  \partial_t \phi_\mA  + {\rm div}(\phi_\mA  \bv) + \rho_\mA^{-1}{\rm div}\hat{\mathbf{H}}_\mA - \rho_\mA^{-1}\hat{\zeta}_\mA &=~0,\label{eq: eq: model alt: mass}
  \end{align}
\end{subequations}
for $\mA = 1,...,N-1$. We note that the combination of the mixture mass balance \eqref{eq: eq: model alt: mass full} and the $N-1$ constituent balance laws \eqref{eq: eq: model alt: mass} are equivalent to the $N$ balance laws \eqref{eq: eq: model simplified: mass} or $N$ balance laws \eqref{eq: eq: model alt: mass}.

While we refrain from discussing formulations that adopt concentration variables, we discuss a formulation in terms of the volume-averaged velocity $\bu$. Inserting the constitutive model for the peculiar velocities \eqref{eq: model J} into \cref{lem: mass av vs vol av} we obtain:
\begin{align}
    \bv = \bu + \rho^{-1}\displaystyle\sum_\mB \hat{\bJ}_\mB^u.
\end{align}
By substituting this identity, we express the model using the volume-averaged velocity:
\begin{subequations}\label{eq: model simplified}
  \begin{align}
   \partial_t \left(\rho\bu + \sum_\mB \hat{\bJ}_\mB^u\right) + {\rm div} \left( \rho \bu\otimes \bu + \sum_\mB \hat{\bJ}_\mB^u \otimes \bu  \right.& \nn\\
              + \left.  \bu\otimes \sum_\mB \hat{\bJ}_\mB^u  + \rho^{-1} \sum_\mB \hat{\bJ}_\mB^u\otimes  \sum_\mB \hat{\bJ}_\mB^u\right) 
   + \nabla \lambda + \sum_{\mB} \phi_\mB\nabla\hat{\mu}_\mB & \nn\\
    - {\rm div} \left(   \nu \left(2\nabla^s\left(\bu +\rho^{-1} \sum_\mB \hat{\bJ}_\mB^u \right)+\bar{\lambda}{\rm div}\left(\bu +\rho^{-1} \sum_\mB \hat{\bJ}_\mB^u \right) \mathbf{I}\right) \right)&=~ 0, \label{eq: model  simplified: mom 2}\\
 \partial_t \phi_\mA  + {\rm div}\left(\phi_\mA \bu\right)+ {\rm div}\hat{\bh}_\mA^u + \rho_\mA^{-1}{\rm div}\hat{\bj}_\mA  -\rho_\mA^{-1}\hat{\zeta}_\mA &=~0.\label{eq: eq: model simplified: mass 2}
  \end{align}
\end{subequations}
Recalling \eqref{eq: mix mass laws: 4}, we have:
\begin{align}
    {\rm div} \bu = \sum_{\mA} \rho_\mA^{-1}(\hat{\zeta}_\mA - {\rm div} \hat{\mathbf{j}}_\mA),
\end{align}
where the right-hand side vanishes when densities match or mass transfer is absent.
Arguably, the formulation \eqref{eq: model simplified} is rather involved. We discuss a simplification in the next subsection.

\subsection{Matching velocities}\label{sec: Matching velocities}

We consider the case in which the peculiar velocities are zero; taking $\mathbf{M}_{\mA\mB} = 0$, we find:
\begin{align}
    \hat{\bJ}_\mA =\hat{\bh}_\mA = \hat{\bJ}_\mA^u =    \hat{\bh}_\mA^u = 0.
\end{align}
As a consequence the mass-averaged
and volume-averaged velocities are equal:
\begin{align}
    \bv = \bu.
\end{align}
This choice models the situation where the constituent velocities are matching. We explicitly state the simplified formulations of the model:
\begin{subequations}\label{eq: model v is u}
  \begin{align}
   \partial_t (\rho \bv) + {\rm div} \left( \rho \bv\otimes \bv \right) + \nabla \lambda + \sum_\mB\phi_\mB\nabla\hat{\mu}_\mB &\nn\\ - {\rm div} \left(   \nu (2\nabla^s \bv+\bar{\lambda}({\rm div}\bv) \mathbf{I}) \right)-\rho\mathbf{g} &=~ 0, \label{eq: model v is u: mom}\\
  \partial_t \phi_\mA  + {\rm div}(\phi_\mA  \bv) + \rho_\mA^{-1} {\rm div}\bj_\mA  - \rho_\mA^{-1}\hat{\zeta}_\mA &=~0.\label{eq: model v is u: mass}
  \end{align}
\end{subequations}
and obviously:
\begin{subequations}\label{eq: model v is u 2}
  \begin{align}
   \partial_t (\rho \bu) + {\rm div} \left( \rho \bu\otimes \bu \right) + \nabla \lambda + \sum_\mB\phi_\mB\nabla\hat{\mu}_\mB &\nn\\- {\rm div} \left(   \nu (2\nabla^s \bu+\bar{\lambda}({\rm div}\bu) \mathbf{I}) \right)-\rho\mathbf{g} &=~ 0, \label{eq: model v is u 2: mom}\\
  \partial_t \phi_\mA  + {\rm div}(\phi_\mA  \bu) + \rho_\mA^{-1} {\rm div}\bj_\mA  - \rho_\mA^{-1}\hat{\zeta}_\mA &=~0.\label{eq: model v is u 2: mass}
  \end{align}
\end{subequations}
The formulation \eqref{eq: model v is u 2} demonstrates that a simplified -- consistent -- model in terms of the volume-averaged velocity involves a straightforward momentum equation. We emphasize that the volume-averaged velocity $\bu$ is in general not divergence-free (recall \eqref{eq: mix mass laws: 4}).

\subsection{Equilibrium conditions}\label{sec: Properties: equi}
We utilize formulation \eqref{eq: model full} to study equilibrium properties. We characterize the set equilibrium solutions $\left\{\mathbf{q}_E=(\bv_E,\phi_{\mA,E},\lambda_E,\mu_{\mA,E})\right\}$ of  \eqref{eq: model full} as stationary solutions subject to boundary conditions for which the dissipation vanishes: $\mathscr{D}(\mathbf{q}_E) = 0$.
Invoking \eqref{eq: second law red 7} yields the conditions:
\begin{subequations}\label{eq: equilibrium cond}
\begin{align}
  \left( \nabla^s \bv_E - \frac{1}{d} ({\rm div} \bv_E) \mathbf{I}\right):\left(\nabla^s \bv_E - \frac{1}{d} ({\rm div} \bv_E) \mathbf{I}\right)=&~0,\label{eq: equi 1}\\
  \left(\lambda_E + \frac{2}{d}\right)\left({\rm div} \bv_E\right)^2 =&~0,\label{eq: equi 2}\\
    \displaystyle\sum_{\mA,\mB} \left(\nabla g_{\mA,E}\right)^T \mathbf{B}_{\mA\mB,E} \nabla g_{\mB,E} =&~0,\label{eq: equi 3}\\
     \sum_{\mA,\mB} m_{\mA\mB,E}g_{\mA,E}g_{\mB,E} =&~0, \label{eq: equi 4}
\end{align}
\end{subequations}
in $\Omega$, where $\mathbf{B}_{\mA\mB}=\mathbf{M}_{\mA\mB} + \mathbf{K}_{\mA\mB}$, and where the subscript $E$ denotes the equilibrium configuration of the quantity. We deduce from \eqref{eq: equi 1}-\eqref{eq: equi 2} that $\bv_E$ are rigid body motions. Simplifying the analysis, we take $\bv_E = 0$, which causes the inertia terms to vanish. Additionally, we assume the absence of gravitational forces ($\mathbf{g} = 0$). Substituting into the momentum equation \eqref{eq: model full: mom} provides:
\begin{align}
  \nabla \lambda_E + \sum_\mA\phi_{\mA,E}\nabla \hat{\mu}_{\mA,E} = 0.
\end{align}
Recalling \eqref{eq: identity lambda mu} we deduce:
\begin{align}\label{eq: constraint 1}
  \sum_\mA\trho_{\mA,E}\nabla g_{\mA,E} = 0.
\end{align} 
Next, from \eqref{eq: equi 3} we deduce that $\nabla g_{E}$ lies in the null space of $\mathbf{B}_{E}$ in the sense $\sum_\mB \mathbf{B}_{\mA\mB,E}\nabla g_{\mB,E} = 0$ ($\mA = 1,...,N$), and hence $\hat{\bJ}_{\mA,E} + \hat{\bj}_{\mA,E} = 0$ in equilibrium.
In the special case $\mathbf{B}_{\mA\mB}=-\mathbf{B}_0 \trho_\mA \trho_\mB$ for $\mA \neq \mB$, and $\mathbf{B}_{\mA\mA} = \mathbf{B}_{0}\trho_\mA \sum_{\mC\neq \mA} \trho_\mC$ for some $\mathbf{B}_0$ that does not depend on $\trho_\mA$, $\mA=1,...,N$, this coincides with \eqref{eq: constraint 1}.
Similarly, \eqref{eq: equi 4} provides $\sum_\mB m_{\mA\mB,E} g_{\mB,E} = 0$, and hence $\hat{\zeta}_{\mA,E}=0$ ($\mA = 1,...,N$).

\section{Connections to existing models}\label{sec: Connections}

This section provides connections with existing models. First, we discuss the binary-phase situation in \cref{sec: Binary constituent case}. Next, in \cref{sec: Dong} we compare the framework with the model of \cite{dong2018multiphase}. Finally, \cref{sec: Class-II mixture model} discusses the link to a model with $N$-momentum equations.

\subsection{Binary-phase case}\label{sec: Binary constituent case}
In this section we restrict to binary mixtures ($\mA = 1,2)$, and compare with the framework presented in \cite{eikelder2023unified}. A formulation of this two-phase modeling framework is:
\begin{subequations}\label{eq: model binary 2}
  \begin{align}
   \partial_t (\rho \bv) + {\rm div} \left( \rho \bv\otimes \bv \right) + \nabla \lambda + \phi \nabla\breve{\mu}
   - {\rm div} \left(   \nu (2\nabla^s \bv+\bar{\lambda}({\rm div}\bv) \mathbf{I}) \right)-\rho\mathbf{g} &=~ 0, \label{eq: model binary 2: mom}\\
   \partial_t \rho + {\rm div}(\rho \bv) &=~0.\\
  \partial_t \phi  + {\rm div}(\phi  \bv) - {\rm div} \left( \breve{\mathbf{M}} \nabla \left(\breve{\mu}+\omega \lambda\right)\right) +\breve{m} \left(\breve{\mu}+\omega \breve{\lambda}\right) &=~0. \label{eq: model binary 2: mass 1}
  \end{align}
\end{subequations}
Here $\phi$ is the phase-field quantity defined as the difference between volume fractions:
\begin{align}\label{eq: def phi}
    \phi = \phi_1 - \phi_2,
\end{align}
where we recall $\phi_1+\phi_2=1$. The chemical potential quantity is defined as:
\begin{align}\label{eq: def mu}
    \breve{\mu} = \dfrac{\partial \Psi}{\partial \phi} - {\rm div}\dfrac{\partial \Psi}{\partial \nabla \phi}.
\end{align}
Finally, the quantity $\omega$ is  $\omega = (\rho_1^{-1}-\rho_2^{-1})/(\rho_1^{-1}+\rho_2^{-1})$. On the other hand, the model \eqref{eq: model full} takes for binary mixtures the following form:
\begin{subequations}\label{eq: model binary}
  \begin{align}
   \partial_t (\rho \bv) + {\rm div} \left( \rho \bv\otimes \bv \right) + \nabla \lambda + \phi_1\nabla\hat{\mu}_1 + \phi_2\nabla\hat{\mu}_2 &\nn\\
   - {\rm div} \left(   \nu (2\nabla^s \bv+\bar{\lambda}({\rm div}\bv) \mathbf{I}) \right)-\rho\mathbf{g} &=~ 0, \label{eq: model binary: mom}\\
  \partial_t \phi_1  + {\rm div}(\phi_1  \bv) - \rho_1^{-1}{\rm div} \left( \mathbf{M} \nabla \left(g_1 -g_2\right)\right) +\rho_1^{-1}m \left(g_1 -g_2\right) &=~0, \label{eq: model binary: mass 1}\\
    \partial_t \phi_2  + {\rm div}(\phi_2  \bv) - \rho_2^{-1}{\rm div} \left( \mathbf{M} \nabla \left(g_2 -g_1\right)\right) +\rho_2^{-1} m \left(g_2 -g_1\right) &=~0,\label{eq: model binary: mass 2}
  \end{align}
\end{subequations}
where $\mathbf{M} = \mathbf{M}_{12}=\mathbf{M}_{21}$ and $m=m_{12}=m_{21}$. By means of variable transformation, we aim to express the model in terms of the quantities of model \eqref{eq: model binary 2}. 
The mass balance equations \eqref{eq: model binary: mass 1} and \eqref{eq: model binary: mass 2} can be written as:
\begin{subequations}\label{eq: model binary mass}
  \begin{align}
  \partial_t \phi  + {\rm div}(\phi  \bv) - {\rm div} \left( (\rho_1^{-1}+\rho_2^{-1})\mathbf{M} \nabla \left(g_1 -g_2\right)\right) 
  + (\rho_1^{-1}+\rho_2^{-1})m \left(g_1 -g_2\right) &=~0, \label{eq: model binary mass: mass 1}\\
    \partial_t \rho  + {\rm div}(\rho  \bv) &=~0.\label{eq: model binary mass: mass 2}
  \end{align}
\end{subequations}
With the aim of comparing the two models, we select the relations $\breve{\mathbf{M}}=(\rho_1^{-1}+\rho_2^{-1})^2\mathbf{M}$ and $\breve{m}=(\rho_1^{-1}+\rho_2^{-1})^2m$, which converts \eqref{eq: model binary mass: mass 1} into:
  \begin{align}
  \partial_t \phi  + {\rm div}(\phi  \bv) - {\rm div} \left( \breve{\mathbf{M}} \nabla \left(\grave{\mu} +\omega \lambda \right)\right) +\breve{m} \left(\grave{\mu}+\omega \lambda \right) &=~0, \label{eq: model binary mass: mass 3}
  \end{align}
with $\grave{\mu} = (\rho_1^{-1}+\rho_2^{-1})^{-1}(\rho_1^{-1}\mu_1-\rho_2^{-1}\mu_2)$. As a consequence, in case the following identities hold:
\begin{subequations}
    \begin{align}
        \phi_1 \nabla \hat{\mu}_1 + \phi_2 \nabla \hat{\mu}_2 =&~ \phi \nabla \breve{\mu},\\
         \grave{\mu} =&~\breve{\mu},
    \end{align}
\end{subequations}
we find that \eqref{eq: model binary} coincides with \eqref{eq: model binary 2}. This is in general not the case, i.e. \textit{in general the two models do not match}\footnote{An $N$-phase theory that reduces to existing two-phase models emerges when working with $N-1$ order parameters $\phi_\mA$, rather than the current case of $N$ order parameters $\phi_\mA$}. There are however specific situations in which the models coincide, for instance when $\mu_1 + \mu_2 = 0$ and $\breve{\mu} = \mu_1 = -\mu_2$. These conditions are inspired by the chain rule for chemical potentials, where $\phi=\phi(\phi_1,\phi_2)=\phi_1-\phi_2$ so that $\partial \phi/\partial \phi_1 = 1$, $\partial \phi/\partial \phi_2 = -1$ of \eqref{eq: def phi}. 

\subsection{$N$-phase model Dong (2018)}\label{sec: Dong}
The $N$-phase incompressible model proposed by \cite{dong2018multiphase} is given by
\begin{subequations}\label{eq: model Dong}
    \begin{align}
        \rho \left(\partial_t \mathbf{u} + \mathbf{u}\cdot \nabla \mathbf{u} \right) + \bJ'\cdot \nabla \bu  + \nabla \lambda' - {\rm div}\left( \nu' \nabla^s \bu \right) &\nn\\
        + \sum_\mB {\rm div}\left(\nabla \phi_\mA \otimes \frac{\partial \Psi'}{\partial (\nabla \phi_\mB)}\right) =&~0,\\
        {\rm div}\bu =&~0,\\
        \partial_t \phi_\mA + \bu\cdot \nabla\phi_\mA - \sum_\mB {\rm div}\left( m_{\mA\mB}' \nabla \left( \frac{\partial \Psi'}{\partial \phi_\mB} - {\rm div}\left(\frac{\partial \Psi'}{\partial (\nabla \phi_\mB)}\right)\right) \right) =&~0,
    \end{align}
\end{subequations}
for $\mA = 1,..., N$, where $\lambda'$ is the Lagrange multiplier pressure, $\nu'$ is the dynamic viscosity, $\Psi'$ is the free energy, $\bJ'$ is the peculiar velocity, and $m_{\mA\mB}'$ is the mobility. For the purpose of comparing the model \eqref{eq: model Dong} to the proposed framework, we define the chemical potential:
\begin{align}
    \mu_\mB' = \frac{\partial \Psi'}{\partial \phi_\mB} - {\rm div}\left(\frac{\partial \Psi'}{\partial (\nabla \phi_\mB)}\right).
\end{align}
Invoking \cref{lem: identity Korteweg stresses}, we rewrite the model \eqref{eq: model Dong} as:
\begin{subequations}\label{eq: dong}
    \begin{align}
        \rho \left(\partial_t \mathbf{u} + \mathbf{u}\cdot \nabla \mathbf{u} \right) + \bJ'\cdot \nabla \bu  + \nabla \tilde{\lambda}' + \sum_\mB \phi_\mA \nabla \mu_\mA' - {\rm div}\left( \nu' \nabla^s \bu \right) =&~0,\\
        {\rm div}\bu =&~0,\\
        \partial_t \phi_\mA + \bu\cdot \nabla\phi_\mA - \sum_\mB {\rm div}\left( m_{\mA\mB}' \nabla \mu_\mB' \right) =&~0,
    \end{align}
\end{subequations}
with
\begin{align}
    \tilde{\lambda}' = \lambda' + \Psi' - \sum_\mB \phi_\mA \mu_\mA'. 
\end{align}
This model \eqref{eq: dong} is not compatible with the framework proposed in the current paper. In particular, comparing \eqref{eq: dong} with \eqref{eq: model simplified}, we observe that:
\begin{itemize}
    \item model \eqref{eq: dong} does not contain each of the peculiar velocity terms in the momentum equation; this applies to both inertia and viscous terms;
    \item model \eqref{eq: dong} does not include mass transfer terms;
    \item the constitutive model for the diffusive flux in \eqref{eq: dong} is different; in particular the Lagrange multiplier is absent. As a consequence, the equilibrium conditions are different.
\end{itemize}

\subsection{Class-II mixture model}\label{sec: Class-II mixture model}
We compare the proposed unified modeling framework with an incompressible mixture model presented in \cite{eikelder2023thermodynamically}:
\begin{subequations}\label{eq: model FE intro}
  \begin{align}
 \partial_t \trho_\mA + {\rm div}(\trho_\mA \bv_\mA) + \sum_\mB \breve{m}_{\mA\mB}(\breve{g}_\mA-\breve{g}_\mB) &=~ 0, \label{eq: model FE: cont intro}\\
 \partial_t (\trho_\mA \bv_\mA) + {\rm div} \left( \trho_\mA \bv_\mA\otimes \bv_\mA \right)  + \phi_\mA\nabla \left(\breve{\lambda} + \breve{\mu}_\mA \right)  \nn\\
 - {\rm div}\left(\breve{\nu}_\mA \left(2\nabla^s \bv_\mA + \breve{\bar{\lambda}}_\mA {\rm div}\mathbf{v}_\mA\right)\right) -\trho_\mA\mathbf{b}&\nn\\
 +\displaystyle\sum_{\mB} R_{\mA\mB} (\bv_\mA-\bv_\mB) + \tfrac{1}{2}\sum_\mB \breve{m}_{\mA\mB}(\breve{g}_\mA-\breve{g}_\mB)(\bv_\mA+\bv_\mB) &=~ 0, \label{eq: model FE: mom intro}
  \end{align}
\end{subequations}
for constituents $\mA = 1,...,N$. Here $\bv_\mA$ is the constituent velocity, $\breve{\lambda}$ is a Lagrange multiplier,$\tilde{\nu}_\mA$ the constituent dynamical viscosity, $\breve{\bar{\lambda}}_\mA\geq 2/d$, $\nabla^s \bv_\mA$ the constituent symmetric velocity gradient, and $\breve{m}_{\mA\mB}$ and $R_{\mA\mB}$ are symmetric matrices (for the properties see \cite{eikelder2023thermodynamically}). This model considers the free energy class:
\begin{subequations}\label{eq: free energy split}
    \begin{align}
    \Psi =&~ \sum_\mA \breve{\Psi}_\mA,\\
    \breve{\Psi}_\mA =&~ \breve{\Psi}_\mA\left(\phi_\mA, \nabla \phi_\mA\right). 
\end{align}
\end{subequations}
The associated constituent chemical potentials are defined as:
\begin{align}
    \breve{\mu}_\mA =&~ \dfrac{ \partial \breve{\Psi}_\mA}{\partial \phi_\mA} - {\rm div}\dfrac{\partial \breve{\Psi}_\mA}{\partial\nabla \phi_\mA},
\end{align}
and $\breve{g}_\mA = \rho_\mA^{-1}(\breve{\mu}_\mA + \breve{\lambda})$. 

Inserting the class \eqref{eq: free energy split} into the proposed modeling framework, we find $\breve{\mu}_\mA=\hat{\mu}_\mA$. Additionally, we identify $\breve{\lambda}=\lambda$; consequently $\breve{g}_\mA=g_\mA$. In contrast to the unified modeling framework presented in the current paper, this model is comprised of $N$ mass balance equations, and $N$ momentum balance equations. As such, we compare the $N$ mass balance laws, and the single mixture momentum balance law of the models. Starting with the mass balance laws, \eqref{eq: model FE: cont intro} can be written as:
\begin{subequations}\label{eq: mass balance}
  \begin{align} 
 \partial_t \phi_\mA + {\rm div}(\phi_\mA \bv) + \rho_\mA^{-1}{\rm div}\bJ_\mA+ \rho_\mA^{-1}\breve{\gamma}_\mA &=~ 0,\\
 \breve{\gamma}_\mA - \sum_\mB \breve{m}_{\mA\mB}(\breve{g}_\mA-\breve{g}_\mB)&=~0.
  \end{align}
\end{subequations}
This form is very similar to \eqref{eq: model full: mass}; the key difference is that the peculiar velocity $\bJ_\mA$ governed by a constitutive model $\bJ_\mA= \hat{\bJ}_\mA$ in the current paper, whereas in \eqref{eq: model FE intro} it follows from the constitutive velocities. With the identification $m_{\mA\mB} = - \breve{m}_{\mA\mB}$ for $\mA\neq\mB$ and $m_{\mA\mB} = \sum_{\mC\neq\mA} \breve{m}_{\mA\mC}$ for $\mA = \mB$ (similar to \cref{rmk: Related constitutive models}) the mass transfer terms match (except for the difference $\hat{\bj}_\mA=0$ in \eqref{eq: model full: mass}). Focusing on the momentum balance laws, addition of \eqref{eq: model FE: mom intro} provides:
\begin{align}\label{eq: sum mom}
 \partial_t (\rho \bv) + {\rm div} \left( \rho\bv\otimes \bv  \right)  + \sum \phi_\mA\nabla \left(\breve{\lambda} + \breve{\mu}_\mA \right)&  \nn\\
 - {\rm div}\left(\sum_\mA\breve{\nu}_\mA \left(2\nabla^s \bv + \breve{\bar{\lambda}}_\mA {\rm div}\bv\right)\right) -\rho\mathbf{b}&\nn\\
 - {\rm div}\left(\sum_\mA\breve{\nu}_\mA \left(2\nabla^s \bw_\mA + \breve{\bar{\lambda}}_\mA {\rm div}\bw_\mA\right)- \sum_\mA \trho_\mA \bw_\mA \otimes \bw_\mA\right)&~=0,
  \end{align}
where we have adopted the identities:
\begin{subequations}
  \begin{align}
    \sum_\mA \trho_\mA \bv_\mA \otimes \bv_\mA =&~ \rho \bv \otimes \bv + \sum_\mA \trho_\mA \bw_\mA \otimes \bw_\mA,\\    
    \sum_\mA\breve{\nu}_\mA \left(2\nabla^s \bv_\mA + \breve{\bar{\lambda}}_\mA {\rm div}\bv_\mA\right) =&~ \sum_\mA\breve{\nu}_\mA \left(2\nabla^s \bv + \breve{\bar{\lambda}}_\mA {\rm div}\bv\right) \nn\\
    &~+ \sum_\mA\breve{\nu}_\mA \left(2\nabla^s \bw_\mA + \breve{\bar{\lambda}}_\mA {\rm div}\bw_\mA\right).
  \end{align}
\end{subequations}
With the identifications $\nu = \sum_\mA \breve{\nu}_\mA$ and $\bar{\lambda}=\breve{\bar{\lambda}}_\mA$ the first two lines match the momentum equation \eqref{eq: model full: mom}. The last line in \eqref{eq: sum mom} consists of terms that are absent in \eqref{eq: model full: mom}. This is a direct consequence of energy-dissipation law \eqref{eq: energy dissipation} and the introduction of the model $\bJ_\mA=\hat{\bJ}_\mA$ in \eqref{eq: class J}. In the case of matching constitutive velocities, as described in \cref{sec: Matching velocities}, these terms vanish.

\section{Conclusion and outlook}\label{sec: discussion} 

This paper presents a unified framework for $N$-phase Navier-Stokes Cahn-Hilliard Allen-Cahn mixture models with non-matching densities. The framework finds its roots in continuum mixture theory, which serves as a fundamental guiding principle for designing multi-physics models at large. The unified framework proposes a (phase-field) system of $N$ mass balance laws, and $1$ momentum balance law, that is invariant to the set of fundamental variables, has an energy-dissipative structure, is reduction-consistent, symmetric with respect to the numbering of the phases, and provides well-defined equilibrium solutions. More specifically, we draw the following conclusions:
\begin{itemize}
    \item The form of the balance laws is invariant to the set of fundamental variables; at both the constituent and mixture levels (\cref{sec: const BL,sec: mixture BL}).
    \item The free energy class depends on all volume fractions (and their gradients) (\cref{sec: const mod: subsec: def,sec: const mod: subsec: model restr}); this provides symmetry with respect to the numbering of the constituents.
    \item Chemical potentials are tightly connected to the Lagrange multiplier that enforces volume conservation; these quantities occur only as superposition (\cref{sec: const mod: subsec: model restr}).
    \item The unified framework is invariant to the set of independent variables; both before and after constitutive modeling (\cref{sec: const mod: alt class,sec: const mod: subsec: select}).
    \item Constitutive quantities are such that the resulting model exhibits energy-dissipation (\cref{sec: const mod: subsec: select}).
    \item Consistency with the single-phase equations requires mobility quantities to be degenerate (\cref{sec: const mod: subsec: select}).
    \item Equilibrium solutions are determined by a balance of (generalized) chemical potentials (see \cref{sec: Properties: equi}).
    \item In the binary case, the framework does, in general, not coincide with existing two-phase models (see \cref{sec: Binary constituent case}). Furthermore, the framework is closely connected to a class-II model (see \cref{sec: Class-II mixture model}), and the model of \cite{dong2018multiphase} does not fit into the framework (see \cref{sec: Dong}).
\end{itemize} 

While the proposed unified framework offers insight into the modeling of $N$-phase flows, we do not claim that it is complete. Therefore, we delineate potential future research directions. First, it is important to study the implications of the particular form of the free energy model, such as equilibrium characteristics, and Ostwald ripening phenomena (see e.g. \cite{ten2024ostwald}). To this purpose, we acknowledge the existence of numerous $N$-phase free energy closure models (see e.g. \cite{boyer2014hierarchy}. Second, it is essential to investigate the sharp interface asymptotic behavior (e.g. jump conditions at interfaces) for particular closure models.
The last point concerns the design of (property-preserving) numerical schemes. {\color{red} Details of $N$-phase computations will be presented elsewhere; however, we provide some considerations here. First, a numerical simulation requires specification of the free energy (as mentioned above). It is hereby important to take \eqref{eq: split free energy} into account to ensure applicability to a general number of constituents. A second consideration concerns the choice of fundamental variables. Although the framework remains invariant to the choice of variables (e.g., using a mass-averaged \eqref{eq: intro mass} or volume-averaged velocity \eqref{eq: intro volume}), certain selections may be more advantageous for designing property-preserving numerical methods. Next, although the proposed system is fully symmetric with respect to the set of variables, in the numerical solution there are at least two roads one can pursue; (i) work with $N-1$ volume fractions and compute the $N$-th volume fraction from the others, or (ii) work with the full set of volume fractions and enforce the saturation constraint. In the second case, the saturation constraint could be enforced via \eqref{eq: mix mass laws: 3}, so that the system of equations becomes:
\begin{subequations}
    \begin{align}
   \partial_t (\rho \bv) + {\rm div} \left( \rho \bv\otimes \bv \right) + \sum_\mB \phi_\mB \nabla (\mu_\mB + \lambda)
    - {\rm div} \left(   \nu (2 \nabla^s \bv+\bar{\lambda}({\rm div}\bv) \mathbf{I}) \right)-\rho\mathbf{b} &=~ 0,\\
  \partial_t \phi_\mA  + {\rm div}(\phi_\mA  \bv) +\rho_\mA^{-1}{\rm div} (\hat{\bH}_\mA )  -\rho_\mA^{-1} \hat{\zeta}_\mA&=~0,\\
    {\rm div}\bv + \displaystyle\sum_{\mB}  \rho_\mB^{-1} \nabla \cdot \hat{\bH}_\mB  - \displaystyle\sum_{\mB}\rho_\mB^{-1}\hat{\zeta}_\mB &~= 0,
  \end{align}
\end{subequations}
with $\hat{\bH}_\mA =- \sum_\mB \mathbf{B}_{\mA\mB}\nabla g_\mB$, $\hat{\zeta}_\mA =- \sum_\mB m_{\mA\mB} g_\mB$, where the mass-averaged velocity is adopted. When working with the mass-averaged velocity, the terms $\hat{\bH}_\mA = \hat{\bJ}_\mA + \hat{\bj}_\mA$ may be modeled together rather than determining $\hat{\bJ}_\mA$ and $\hat{\bj}_\mA$ independently. In contrast, within the volume-averaged velocity formulation of the model, these terms serve a distinct role. Taking $\bj_\mA = 0$ and $\zeta_\mA = 0$, $\mA = 1,...,N$ then provides a divergence-free velocity. Finally, the model can accommodate large differences in specific densities between constituents. Ensuring this property in the fully discrete case requires a robust numerical method.
}

\appendix

\section{Reduced free energy class, and proofs}\label{appendix proofs}

We briefly discuss the free energy class with reduced dependency:
\begin{align}\label{eq: class Psi alt}
  \Psi =&~ \doublehat{\Psi}^{(\beta)}\left(\left\{\phi_\mA\right\}_{\mA\neq\mB},\left\{\nabla \phi_\mA\right\}_{\mA\neq\mB}\right),
\end{align}
where the constituent number $\beta \in \left\{1,...,N\right\}$ is fixed, and where both $\left\{\phi_\mA\right\}_{\mA\neq\mB}$ and $\left\{\nabla \phi_\mA\right\}_{\mA\neq\mB}$ consist of independent variables. The class \eqref{eq: class Psi alt} is connected to \eqref{eq: class Psi} via the identification:
\begin{align}\label{eq: class Psi 2}
\Psi=&~\hat{\Psi}\left(\left\{\phi_\mA\right\}_{\mA\neq\mB},1-\sum_{\mA\neq\mB}\phi_\mA,\left\{\nabla \phi_\mA\right\}_{\mA\neq\mB},-\sum_{\mA\neq\mB}\nabla\phi_\mA\right)\nn\\
=&~
\doublehat{\Psi}^{(\beta)}\left(\left\{\phi_\mA\right\}_{\mA\neq\mB},\left\{\nabla \phi_\mA\right\}_{\mA\neq\mB}\right).
\end{align}
The associated chemical potentials take the form:
\begin{align}
    \doublehat{\mu}_\mA^{(\beta)} = \dfrac{ \partial \doublehat{\Psi}^{(\beta)}}{\partial \phi_\mA} - {\rm div}\dfrac{\partial \doublehat{\Psi}^{(\beta)}}{\partial\nabla \phi_\mA}.
\end{align}

\begin{lemma}[Chemical potentials reduced class]\label[lemma]{appendix: lemma chemical potentials reduced class}
The chemical potentials of the reduced class may be expressed as
  \begin{align}
    \doublehat{\mu}_\mA^{(\beta)}  = \hat{\mu}_\mA - \hat{\mu}_\mB.
  \end{align}
\end{lemma}
\begin{proof}
Direct evaluation of the partial derivatives provides:
\begin{subequations}\label{eq: der expressions 0}
    \begin{align}
  \dfrac{\partial \doublehat{\Psi}^{(\beta)}}{\partial \phi_\mA} =&~ \dfrac{\partial \hat{\Psi}}{\partial \phi_\mA}+\dfrac{\partial \hat{\Psi}}{\partial \phi_\mB}\dfrac{\partial \phi_\mB}{\partial \phi_\mA} = \dfrac{\partial \hat{\Psi}}{\partial \phi_\mA}-\dfrac{\partial \hat{\Psi}}{\partial \phi_\mB},\\
  \dfrac{\partial \doublehat{\Psi}^{(\beta)}}{\partial \nabla\phi_\mA} = &~\dfrac{\partial \hat{\Psi}}{\partial \nabla\phi_\mA}+\dfrac{\partial \hat{\Psi}}{\partial \nabla\phi_\mB}\dfrac{\partial \phi_\mB}{\partial \phi_\mA} = \dfrac{\partial \hat{\Psi}}{\partial \nabla\phi_\mA}-\dfrac{\partial \hat{\Psi}}{\partial \nabla\phi_\mB}.
  \end{align}
The linearity of the divergence operator concludes the proof.
\end{subequations}
\end{proof}
\begin{lemma}[Derivative of the free energy]\label[lemma]{appendix: lem derivative}
The derivative of the free energy class \eqref{eq: class Psi} is given by:
\begin{align}
    {\rm d}\hat{\Psi} = \sum_{\mA} \dfrac{\partial \hat{\Psi}}{\partial \phi_\mA} {\rm d}\phi_\mA + \sum_{\mA} \dfrac{\partial \hat{\Psi}}{\partial \nabla\phi_\mA} {\rm d}(\nabla \phi_\mA),
\end{align}
where ${\rm d}$ is the derivative operator.
\end{lemma}
\begin{proof}
Inserting \eqref{eq: der expressions 0}, the derivative of $\Psi$ takes the form
\begin{align}
    {\rm d}\hat{\Psi} ={\rm d}\doublehat{\Psi}^{(\beta)} =&~ \sum_{\mA\neq\mB} \dfrac{\partial \doublehat{\Psi}^{(\beta)}}{\partial \phi_\mA} {\rm d}\phi_\mA + \sum_{\mA\neq\mB} \dfrac{\partial \doublehat{\Psi}^{(\beta)}}{\partial \nabla \phi_\mA} {\rm d}(\nabla \phi_\mA)\nn\\
    =&~ \sum_{\mA\neq\mB} \left(\dfrac{\partial \hat{\Psi}}{\partial \phi_\mA}-\dfrac{\partial \hat{\Psi}}{\partial \phi_\mB}\right) {\rm d}\phi_\mA + \sum_{\mA\neq\mB} \left(\dfrac{\partial \hat{\Psi}}{\partial \nabla\phi_\mA}-\dfrac{\partial \hat{\Psi}}{\partial \nabla\phi_\mB}\right) {\rm d}(\nabla \phi_\mA)\nn\\
    =&~ \sum_{\mA\neq\mB} \dfrac{\partial \hat{\Psi}}{\partial \phi_\mA} {\rm d}\phi_\mA + \sum_{\mA\neq\mB} \dfrac{\partial \hat{\Psi}}{\partial \nabla\phi_\mA} {\rm d}(\nabla \phi_\mA)    \nn\\
    &~- \dfrac{\partial \hat{\Psi}}{\partial \phi_\mB}\sum_{\mA\neq\mB}  {\rm d}\phi_\mA - \dfrac{\partial \hat{\Psi}}{\partial \nabla\phi_\mB}\sum_{\mA\neq\mB} {\rm d}(\nabla \phi_\mA)\nn\\
    =&~ \sum_{\mA\neq\mB} \dfrac{\partial \hat{\Psi}}{\partial \phi_\mA} {\rm d}\phi_\mA + \sum_{\mA\neq\mB} \dfrac{\partial \hat{\Psi}}{\partial \nabla\phi_\mA} {\rm d}(\nabla \phi_\mA)    + \dfrac{\partial \hat{\Psi}}{\partial \phi_\mB}  {\rm d}\phi_\mB + \dfrac{\partial \hat{\Psi}}{\partial \nabla\phi_\mB}{\rm d}(\nabla \phi_\mB)\nn\\  
    =&~ \sum_{\mA} \dfrac{\partial \hat{\Psi}}{\partial \phi_\mA} {\rm d}\phi_\mA + \sum_{\mA} \dfrac{\partial \hat{\Psi}}{\partial \nabla\phi_\mA} {\rm d}(\nabla \phi_\mA),
\end{align}
where we have invoked $\sum_{\mA}{\rm d}\phi_\mA = 0$ and $\sum_{\mA}{\rm d}(\nabla\phi_\mA) = 0$. The latter expression matches the unconstrained derivative.

\end{proof}

\begin{lemma}[Well-defined free energy terms]\label[lemma]{appendix: lem: well-defined}
The following free energy terms in \eqref{eq: IP} are well-defined:
\begin{align}
  \sum_{\mA} \hat{\mu}_\mA \dot{\phi}_\mA; \quad \sum_{\mA} \nabla \phi_\mA \otimes \dfrac{\partial \hat{\Psi}}{\partial \nabla \phi_\mA}; \quad \sum_{\mA} \dot{\phi}_\mA \dfrac{\partial \hat{\Psi}}{\partial \nabla \phi_\mA}.
\end{align}
\end{lemma}
\begin{proof}
If the constraint \eqref{eq: sum phi} is not enforced, the terms are obviously well-defined. We show that the first term subject to \eqref{eq: sum phi} is well-defined; the others follow similarly. Utilizing an argumentation analogously to that of the proof of \cref{appendix: lem derivative}, we have the sequence of identities:
\begin{align}
    \sum_{\mA} \hat{\mu}_\mA \dot{\phi}_\mA    =&~ \sum_{\mA\neq\mB} \hat{\mu}_\mA \dot{\phi}_\mA + \hat{\mu}_\mB\dot{\phi}_\mB\nn\\
    =&~\sum_{\mA\neq\mB} \hat{\mu}_\mA \dot{\phi}_\mA  - \hat{\mu}_\mB \sum_{\mA\neq\mB} \dot{\phi}_\mA\nn\\
    =&~ \sum_{\mA\neq\mB} \left(\hat{\mu}_\mA-\hat{\mu}_\mB\right) \dot{\phi}_\mA\nn\\
    =&~\sum_{\mA\neq\mB} \doublehat{\mu}_\mA^{(\beta)} \dot{\phi}_\mA,
\end{align}
where we have utilized \cref{appendix: lemma chemical potentials reduced class} in the last identity, where we have invoked $\sum_{\mA}\phi_\mA = 1$. Since the latter expression is well-defined, so is the initial one.
\end{proof}
\begin{lemma}[Free energy identity]\label[lemma]{lem: free energy identity}
The following identity holds:
\begin{align}
\sum_\mA \phi_\mA \nabla \mu_\mA=\nabla \left(\sum_\mA \phi_\mA \hat{\mu}_\mA - \hat{\Psi}\right) + {\rm div} \left(\sum_\mA  \nabla \phi_\mA \otimes \dfrac{\partial \hat{\Psi}}{\partial \nabla\phi_\mA} \right).
\end{align}
\end{lemma}
\begin{proof}
Expanding the derivatives of the right-hand size term yields:
\begin{align}\label{eq: final expression}
&\nabla \left(\sum_\mA \phi_\mA \mu_\mA - \hat{\Psi}\right) + {\rm div} \left(\sum_\mA  \nabla \phi_\mA \otimes \dfrac{\partial \hat{\Psi}}{\partial \nabla\phi_\mA} \right) =\nn\\
&\sum_\mA \phi_\mA \nabla\mu_\mA +\sum_\mA \nabla \phi_\mA \dfrac{ \partial \hat{\Psi}}{\partial \phi_\mA} - \sum_\mA\nabla \phi_\mA \divg\left(\dfrac{\partial \hat{\Psi}}{\partial \nabla\phi_\mA}\right)  - \nabla \hat{\Psi} \nn\\
&+ \sum_\mA\nabla \phi_\mA {\rm div} \left( \dfrac{\partial \hat{\Psi}}{\partial \nabla\phi_\mA} \right)+\sum_\mA\left(\mathbf{H} \phi_\mA\right) \dfrac{\partial \hat{\Psi}}{\partial \nabla\phi_\mA} =\nn\\
&\sum_\mA\phi_\mA \nabla\mu_\mA - \nabla \hat{\Psi} +\sum_\mA\nabla \phi_\mA \dfrac{ \partial \hat{\Psi}}{\partial \phi_\mA}   +\sum_\mA\left(\mathbf{H} \phi_\mA\right) \dfrac{\partial \hat{\Psi}}{\partial \nabla\phi_\mA},
\end{align}
where $\mathbf{H} \phi_\mA$ is the hessian of $\phi_\mA$. Observing that the sum of the latter three terms in the final expression in \eqref{eq: final expression} vanishes completes the proof.
\end{proof}

\section{Equivalence of modeling restrictions}\label{appendix: sec: Equivalence of modeling restrictions}
This section discusses the equivalence of the restrictions \eqref{eq: second law 4} and \eqref{eq: restriction alt} via variable transformation.

First, we recall the variable transformation \eqref{eq: relation c phi}:
\begin{subequations}\label{appendix: eq: relation c phi}
    \begin{align}
        \phi_\mA =&~ \frac{c_\mA}{\rho_\mA}\left(\sum_\mB\frac{c_\mB}{\rho_\mB}\right)^{-1} = \frac{c_\mA}{\rho_\mA}\rho,\\
        c_\mA =&~ \rho_\mA\phi_\mA  \left(\sum_\mB\rho_\mB\phi_\mB\right)^{-1} = \rho_\mA\phi_\mA  \rho^{-1},
    \end{align}
\end{subequations}
where we note:
\begin{subequations}
   \begin{align}
    \rho = \hat{\rho}\left(\left\{c_\mB\right\}\right)&~ \left(\sum_\mB\frac{c_\mB}{\rho_\mB}\right)^{-1},\\
    \rho^{-1}=\check{\rho}^{-1}\left(\left\{\phi_\mB\right\}\right)&~ =\left(\sum_\mB\rho_\mB\phi_\mB\right)^{-1},
\end{align}
\end{subequations}
for $\mA =1, ..., N$.

\begin{lemma}[Invertibility transformation maps]\label[lemma]{appendix; lem transformation maps}
    The maps \eqref{appendix: eq: relation c phi} are not invertible.
\end{lemma}
\begin{proof}
  A straightforward evaluation provides the elements of the Jacobian mappings:
\begin{subequations}\label{eq: der expressions}
    \begin{align}
        \dfrac{\partial \phi_\mB}{\partial c_\mA} =&~ \rho \rho_\mB^{-1} \left(\delta_{\mA\mB} - \rho c_\mB \rho_\mA^{-1}\right),\\
        \dfrac{\partial c_\mB}{\partial \phi_\mA} =&~ \rho^{-1} \rho_\mB \left(\delta_{\mA\mB} - \rho^{-1} \phi_\mB \rho_\mA\right),
    \end{align}
\end{subequations}
where $\delta_{\mA\mB}$ is the kroneckerdelta. Summation over $\mB = 1,..., N$ yields:
\begin{subequations}\label{eq: der expressions sums}
    \begin{align}
       \sum_\mB \dfrac{\partial \phi_\mB}{\partial c_\mA} = &~ 0, \\
       \sum_\mB \dfrac{\partial c_\mB}{\partial \phi_\mA} =&~ 0.
    \end{align}
\end{subequations}
Hence, each of the columns of the Jacobian sums to zero. Thus the columns are linearly dependent, and consequently the determinants of the both mappings vanish:
\begin{subequations}\label{eq: determinants}
    \begin{align}
        {\rm det} \dfrac{\partial \phi_\mB}{\partial c_\mA} =&~ 0,\\
        {\rm det} \dfrac{\partial c_\mB}{\partial \phi_\mA} =&~ 0.\\ \nn
    \end{align}
\end{subequations}
\end{proof}

Next, we recall the chain rule for the chemical potential.

\begin{lemma}[Chain rule chemical potentials]\label[lemma]{lem: chain rule chemical potentials}
We have the chain rule for chemical potentials:
\begin{subequations}\label{eq: chain rule}
  \begin{align}
    \check{\mu}_\mA =&~\sum_\mB \hat{\mu}_\mB  \dfrac{\partial \phi_\mB}{\partial c_\mA}, \label{eq: chain rule 1}\\
    \hat{\mu}_\mA =&~\sum_\mB \check{\mu}_\mB  \dfrac{\partial c_\mB}{\partial \phi_\mA}.\label{eq: chain rule 2}
  \end{align}
\end{subequations}
\end{lemma}
\begin{proof}
We show \eqref{eq: chain rule 1} and note that \eqref{eq: chain rule 2} follows similarly. A direct computation yields:
  \begin{align}
    \check{\mu}_\mA =&~  \dfrac{\partial \hat{\Psi}\left(\phi_{\mB}\left(\left\{c_\mC\right\}\right),\sum_\mC \frac{\partial \phi_\mB}{\partial c_\mC}\nabla c_\mC\right)}{\partial c_\mA} - \divg \dfrac{\partial \hat{\Psi}\left(\phi_{\mB}\left(\left\{c_\mC\right\}\right),\sum_\mC \frac{\partial \phi_\mB}{\partial c_\mC}\nabla c_\mC\right)}{\partial \nabla c_\mA}  \nn\\
    =&~ \sum_\mB \dfrac{\partial \hat{\Psi}\left(\left\{\phi_{\mB}\right\},\left\{\nabla \phi_\mB\right\}\right)}{\partial \phi_\mB}\dfrac{\partial \phi_\mB}{\partial c_\mA}\nn\\
    &~+\sum_\mB \dfrac{\partial \hat{\Psi}\left(\left\{\phi_{\mB}\right\},\left\{\nabla \phi_\mB\right\}\right)}{\partial \nabla\phi_\mB}\cdot\left(\sum_\mC \nabla c_\mC~\frac{\partial^2 \phi_\mB}{\partial c_\mA \partial c_\mC}\right)\nn\\
    &~-\sum_\mB  \divg \hat{\Psi}\left(\left\{\phi_{\mB}\right\},\left\{\nabla \phi_\mB\right\}\right) \frac{\partial \phi_\mB}{\partial c_\mA} \nn\\
    &~-\sum_\mB \dfrac{\partial \hat{\Psi}\left(\left\{\phi_{\mB}\right\},\left\{\nabla \phi_\mB\right\}\right)}{\partial \nabla\phi_\mB}\cdot\nabla \left(\frac{\partial \phi_\mB}{\partial c_\mA}\right)\nn\\
    =&~ \sum_\mB \left(\dfrac{\partial \hat{\Psi}\left(\left\{\phi_{\mB}\right\},\left\{\nabla \phi_\mB\right\}\right)}{\partial \phi_\mB}- \divg \left(\dfrac{\partial \hat{\Psi}\left(\left\{\phi_{\mB}\right\},\left\{\nabla \phi_\mB\right\}\right)}{\partial \nabla \phi_\mB}  \right)\right) \dfrac{\partial \phi_\mB}{\partial c_\mA} \nn\\
    =&~ \sum_\mB \hat{\mu}_\mB  \dfrac{\partial \phi_\mB}{\partial c_\mA}.
\end{align}
\end{proof}

\begin{lemma}[Relations between chemical quantities]\label[lemma]{appendix: lem: Relations chemical potential-like quantities}
The chemical potential quantities are related via the following identities:
  \begin{subequations}\label{eq: relation chem}
    \begin{align}
      \hat{\mu}_\mA =&~ \rho^{-1} \rho_\mA\left(\check{\mu}_\mA  -  \sum_\mB \check{\mu}_\mB c_\mB\right),\label{eq: chem pot relation 1}\\
      \check{\mu}_\mA =&~ \rho \rho_\mA^{-1} \left(\hat{\mu}_\mA - \sum_\mB \hat{\mu}_\mB \phi_\mB\right).\label{eq: chem pot relation 2}
    \end{align}
  \end{subequations}
\end{lemma}
\begin{proof}
This follows from substituting \eqref{eq: der expressions} into \eqref{lem: chain rule chemical potentials}.
\end{proof}

\begin{lemma}[Matching Korteweg tensors]\label[lemma]{appendix: lem: matching korteweg tensors}
  The Korteweg stress tensors of the both modeling choices are identical:
  \begin{align}
      \sum_\mA \nabla \phi_\mA \otimes \dfrac{\partial\hat{\Psi}}{\partial \nabla \phi_\mA} = 
      \sum_\mA \nabla c_\mA \otimes \dfrac{\partial\check{\Psi}}{\partial \nabla c_\mA}.
  \end{align}
\end{lemma}
\begin{proof}
  This follows from \eqref{eq: der expressions} and \eqref{eq: sum phi}:
    \begin{align}
       \sum_\mA \nabla c_\mA \otimes \dfrac{\partial\check{\Psi}}{\partial \nabla c_\mA}  =&~ 
       \sum_\mA \left(\sum_\mB \dfrac{\partial c_\mA}{\partial \phi_\mB} \nabla \phi_\mB\right) \otimes \left(\sum_\mC \dfrac{\partial\hat{\Psi}}{\partial \nabla \phi_\mC}\dfrac{\partial \phi_\mC}{\partial c_\mA}\right)\nn\\
       =&~ 
       \sum_{\mA,\mB,\mC} \left( \dfrac{\partial c_\mA}{\partial \phi_\mB} \dfrac{\partial \phi_\mC}{\partial c_\mA}\right) \nabla \phi_\mB \otimes  \dfrac{\partial\hat{\Psi}}{\partial \nabla \phi_\mC}\nn\\
       =&~ 
       \sum_{\mB,\mC} \delta_{\mB\mC} \nabla \phi_\mB \otimes  \dfrac{\partial\hat{\Psi}}{\partial \nabla \phi_\mC}\nn\\
        =&~
       \sum_\mA \nabla \phi_\mA \otimes \dfrac{\partial\hat{\Psi}}{\partial \nabla \phi_\mA}. \\ \nn
  \end{align}
\end{proof}
\begin{theorem}[Equivalence modeling restrictions]\label{appendix: thm: mod restrictions}
  The modeling restrictions \eqref{eq: second law 4} and \eqref{eq: restriction alt} are equivalent.
\end{theorem}
\begin{proof}
We select the following relations between the Lagrange multipliers of the two modeling choices:
\begin{align}\label{appendix: eq: relation LM}
    \check{\lambda}=&~ \hat{\lambda}+\sum_\mB\hat{\mu}_\mB\phi_\mB.
\end{align}
Invoking \cref{appendix: lem: matching korteweg tensors} and substituting the relation \eqref{appendix: eq: relation LM} provides:
\begin{align}\label{appendix: ID1}
     \check{\lambda}\mathbf{I} +\sum_{\mA} \nabla c_\mA \otimes \dfrac{\partial \check{\Psi}}{\partial \nabla c_\mA}  -  \check{\Psi}\mathbf{I} = \hat{\lambda}\mathbf{I} +\sum_{\mA} \nabla \phi_\mA \otimes \dfrac{\partial \hat{\Psi}}{\partial \nabla \phi_\mA} + \left(\hat{\mu}_\mA \phi_\mA -  \check{\Psi}\right)\mathbf{I}.
\end{align}
In a similar fashion we find
\begin{align}
   \rho^{-1}\check{\mu}_\mA + \rho_\mA^{-1}\check{\lambda} =&~ \rho_\mA^{-1}\hat{\lambda}+   \rho_\mA^{-1} \left(\hat{\mu}_\mA - \sum_\mB \hat{\mu}_\mB \phi_\mB\right) + \rho_\mA^{-1}\sum_\mA\hat{\mu}_\mB\phi_\mB= g_\mA,
\end{align}
and conclude:
\begin{subequations}
    \begin{align}
 - \sum_{\mA}\nabla \left( \rho^{-1}\check{\mu}_\mA + \rho_\mA^{-1}\check{\lambda} \right)\cdot \bH_\mA  =&~        - \sum_{\mA}\nabla g_\mA\cdot \bH_\mA \\
 -\sum_{\mA} \left( \rho^{-1}\check{\mu}_\mA + \rho_\mA^{-1}\check{\lambda} \right) \zeta_\mA = &~-\sum_{\mA} g_\mA \zeta_\mA.
    \end{align}
\end{subequations}

\end{proof}

  
Finally, we note that failure to properly account for the saturation constraint \eqref{eq: sum c phi} may result in erroneous derivations. For example, 
from \eqref{eq: relation chem} one can deduce 
  \begin{subequations}\label{eq: problem}
      \begin{align}
          \sum_\mA \hat{\mu}_\mA \phi_\mA =&~ 0,\\
          \sum_\mA \check{\mu}_\mA c_\mA =&~ 0,
      \end{align}
  \end{subequations}
  which do not hold in general. In particular, this follows from the chain rule \cref{lem: chain rule chemical potentials}:
  \begin{subequations}
      \begin{align}
          \sum_\mA \hat{\mu}_\mA \phi_\mA =&~ \sum_\mA \left(\sum_\mB \check{\mu}_\mB \frac{\partial c_\mB}{\partial \phi_\mA} \right)\phi_\mA,\\
          \sum_\mA \check{\mu}_\mA c_\mA =&~\sum_\mA\left(\sum_\mB \hat{\mu}_\mB \frac{\partial \phi_\mB}{\partial c_\mA} \right)c_\mA,
      \end{align}
  \end{subequations}
  alongside with the identities:
  \begin{subequations}
      \begin{align}
          \sum_\mA \dfrac{\partial c_\mB}{\partial \phi_\mA} \phi_\mA =&~ 0,\\
          \sum_\mA \dfrac{\partial \phi_\mB}{\partial c_\mA} c_\mA =&~ 0.
      \end{align}
  \end{subequations}
In this situation we have $\hat{\lambda}= \check{\lambda}$. 

\section{Alternative constitutive modeling}\label{appendix: alt const mod}

In this section we provide some brief details on the constitutive modeling based on concentration variables. \cref{appendix: sec: const mod: subsec: def} outlines the modeling assumptions, and \cref{appendix: sec: const mod: subsec: model restr} derives the modeling restriction.

\subsection{Assumptions and modeling choices}\label{appendix: sec: const mod: subsec: def}

We use the balance laws \eqref{eq: BL constitutive}, where the mass balance laws are now written in terms of concentration variables:
\begin{subequations}\label{appendix: eq: BL constitutive}
  \begin{align}
        \rho(\partial_t c_\mA + \bv\cdot\nabla c_\mA) + {\rm div} \bH_\mA &=~ \zeta_\mA, \label{appendix: eq: BL constitutive: mass} \\
        \partial_t (\rho \bv) + {\rm div} \left( \rho\bv\otimes \bv \right) -  {\rm div} \mathbf{T} -  \rho \mathbf{b} &=~0,\label{appendix: eq: lin mom mix}\\
        \mathbf{T}-\mathbf{T}^T &=~0,\label{appendix: eq: ang mom mix}
  \end{align}
\end{subequations}
where \eqref{appendix: eq: BL constitutive: mass} holds for constituents $\mA=1,...,N$. 
We use the energy-dissipation law \eqref{eq: energy dissipation}:
\begin{align}\label{appendix: eq: energy dissipation}
    \dfrac{{\rm d}}{{\rm d}t} \mathscr{E} = \mathscr{W} - \mathscr{D},
\end{align}
with dissipation $\mathscr{D}\geq 0$, and recall \eqref{eq: total energy}. We postulate the free energy to pertain to the constitutive class:
\begin{align}\label{appendix: eq: class Psi}
  \Psi =&~ \check{\Psi}\left(\left\{c_\mA\right\}_{\mA=1,...,N},\left\{\nabla c_\mA\right\}_{\mA=1,...,N}\right),
\end{align}
subject to the summation constraint \eqref{eq: sum c}, and introduce the chemical potential quantities ($\alpha = 1,...,N$):
\begin{align}
    \check{\mu}_\mA =&~ \dfrac{ \partial \check{\Psi}}{\partial c_\mA} - {\rm div}\dfrac{\partial \check{\Psi}}{\partial\nabla c_\mA}.
\end{align}

\subsection{Modeling restriction}\label{appendix: sec: const mod: subsec: model restr}

By applying Reynolds transport theorem, the divergence theorem and integration by parts, and identity \eqref{eq: relation grad phi}, the evolution of the free energy $\check{\Psi}$ takes the form:
\begin{align}
    \dfrac{{\rm d}}{{\rm d}t}\displaystyle\int_{\mathcal{R}(t)} \check{\Psi} ~{\rm d}v  =&~ \displaystyle\int_{\mathcal{R}(t)} \check{\Psi}~{\rm div} \bv + \sum_{\mA} \check{\mu}_\mA \dot{c}_\mA  - \sum_{\mA} \nabla c_\mA \otimes \dfrac{\partial \check{\Psi}}{\partial \nabla c_\mA}: \nabla \bv  ~{\rm d}v  \nn\\
    &~+ \displaystyle\int_{\partial \mathcal{R}(t)}\sum_{\mA} \dot{c}_\mA \dfrac{\partial \check{\Psi}}{\partial \nabla c_\mA}\cdot \boldsymbol{\nu} ~{\rm d}a.
\end{align}
Substituting the constituent mass balance laws \eqref{eq: BL constitutive: mass}, and again applying integration by parts, yields:
\begin{align}\label{appendix: eq: Psi 2}
    \dfrac{{\rm d}}{{\rm d}t}\displaystyle\int_{\mathcal{R}(t)} \sum_{\mA} \check{\Psi} ~{\rm d}v = &~ \displaystyle\int_{\mathcal{R}(t)} \check{\Psi}~{\rm div} \bv  + \sum_{\mA} \nabla (\rho^{-1}\check{\mu}_\mA) \cdot \bH_\mA \nn\\
    &~~~- \sum_{\mA} \nabla c_\mA \otimes \dfrac{\partial \check{\Psi}}{\partial \nabla c_\mA}: \nabla \bv + \sum_{\mA}  \rho^{-1}\check{\mu}_\mA \zeta_\mA~{\rm d}v\nn\\
    &~+ \displaystyle\int_{\partial \mathcal{R}(t)}\left(\sum_{\mA} \dot{c}_\mA \dfrac{\partial \check{\Psi}}{\partial \nabla c_\mA}-\rho^{-1}\check{\mu}_\mA \bH_\mA\right)\cdot \boldsymbol{\nu} ~{\rm d}a.
\end{align}
Addition of \eqref{eq: kin grav evo} and \eqref{appendix: eq: Psi 2} provides the evolution of the total energy:
\begin{align}\label{appendix: eq: second law subst 1}
    \dfrac{{\rm d}}{{\rm d}t} \mathscr{E} = &~ \displaystyle\int_{\partial \mathcal{R}(t)}\left(\bv^T\mathbf{T}-\sum_{\mA} \left(\rho^{-1}\check{\mu}_\mA\bH_\mA -\dot{c}_\mA \dfrac{\partial \check{\Psi}}{\partial \nabla c_\mA}\right)\right)\cdot \boldsymbol{\nu} ~{\rm d}a \nn\\
    &~- \displaystyle\int_{\mathcal{R}(t)}   \left(\mathbf{T}  +\sum_{\mA} \nabla c_\mA \otimes \dfrac{\partial \check{\Psi}}{\partial \nabla c_\mA} -  \check{\Psi}\mathbf{I}\right):\nabla \mathbf{v}\nn\\
    &~~~~~~~~~~~+ \sum_{\mA}\left(-\nabla (\rho^{-1}\check{\mu}_\mA)\cdot \bH_\mA  - \rho^{-1}\check{\mu}_\mA\zeta_\mA \right)~{\rm d}v.
\end{align}
Analogously to \cref{sec: const mod: subsec: model restr}, we restore the degenerate nature of \eqref{appendix: eq: second law subst 1} via a Lagrange multiplier construction:
\begin{align}\label{appendix: eq: LM 1}
    0=&~ \check{\lambda}  {\rm div}\bv + \nabla \left(\check{\lambda} \displaystyle\sum_{\mA}  \rho_\mA^{-1} \bH_\mA\right) - \sum_{\mA} \rho_\mA^{-1}\bH_\mA \cdot\nabla \check{\lambda} - \check{\lambda} \displaystyle\sum_{\mA}  \rho_\mA^{-1} \zeta_\mA,
\end{align}
where $\check{\lambda}$ is the scalar Lagrange multiplier. 
Integrating \eqref{appendix: eq: LM 1} over $\mathcal{R}(t)$ and subtracting the result from \eqref{appendix: eq: second law subst 1} provides:
\begin{align}\label{appendix: eq: second law subst 2}
    \dfrac{{\rm d}}{{\rm d}t} \mathscr{E} = &~ \displaystyle\int_{\partial \mathcal{R}(t)}\left(\bv^T\mathbf{T}-\sum_{\mA} \left(\rho^{-1}\check{\mu}_\mA\bH_\mA -\dot{c}_\mA \dfrac{\partial \check{\Psi}}{\partial \nabla c_\mA}\right)-\check{\lambda} \displaystyle\sum_{\mA} \rho_\mA^{-1} \bH_\mA\right)\cdot \boldsymbol{\nu} ~{\rm d}a \nn\\
    &~- \displaystyle\int_{\mathcal{R}(t)}   \left(\mathbf{T} + \check{\lambda}\mathbf{I} +\sum_{\mA} \nabla c_\mA \otimes \dfrac{\partial \check{\Psi}}{\partial \nabla c_\mA}  -  \check{\Psi}\mathbf{I}\right):\nabla \mathbf{v}\nn\\
    &~~~~~~~~~~~~~~~~+ \sum_{\mA}\left(-\nabla \left( \rho^{-1}\check{\mu}_\mA + \rho_\mA^{-1}\check{\lambda} \right)\cdot \bH_\mA  \right.\nn\\
    &~~~~~~~~~~~~~~~~~~~~~\left.- \left( \rho^{-1}\check{\mu}_\mA+ \rho_\mA^{-1}\check{\lambda} \right) \zeta_\mA \right)~{\rm d}v.
\end{align}
The rate of work and the dissipation take the forms:
\begin{subequations}\label{appendix: eq: W, D}
\begin{align}
    \mathscr{W} =&~ \displaystyle\int_{\partial \mathcal{R}(t)}\left(\bv^T\mathbf{T}-\sum_{\mA} \left(\rho^{-1}\check{\mu}_\mA\bH_\mA -\dot{c}_\mA \dfrac{\partial \check{\Psi}}{\partial \nabla c_\mA}\right)-\check{\lambda} \displaystyle\sum_{\mA} \rho_\mA^{-1}  \bH_\mA\right)\cdot \boldsymbol{\nu} ~{\rm d}a,\\
    \mathscr{D} =&~\displaystyle\int_{\mathcal{R}(t)}   \left(\mathbf{T} + \check{\lambda}\mathbf{I} +\sum_{\mA} \nabla c_\mA \otimes \dfrac{\partial \check{\Psi}}{\partial \nabla c_\mA}  -  \check{\Psi}\mathbf{I}\right):\nabla \mathbf{v}\nn\\
    &~~~~~~~~~~~~~+ \sum_{\mA}\left(-\nabla \left( \rho^{-1}\check{\mu}_\mA + \rho_\mA^{-1}\check{\lambda} \right)\cdot \bH_\mA \right.\nn\\
    &~~~~~~~~~~~~~~~~~~~ \left.-\left( \rho^{-1}\check{\mu}_\mA + \rho_\mA^{-1}\check{\lambda} \right)\zeta_\mA \right)~{\rm d}v.\label{eq: def diffusion 2}
\end{align}
\end{subequations}
Given that the control volume $\mathcal{R}=\mathcal{R}(t)$, can be chosen arbitrarily, adhering to the energy dissipation law requires that the following local inequality is satisfied:
\begin{align}\label{appendix: eq: second law 4}
     \left(\mathbf{T} + \check{\lambda}\mathbf{I} +\sum_{\mA} \nabla c_\mA \otimes \dfrac{\partial \check{\Psi}}{\partial \nabla c_\mA}  -  \check{\Psi}\mathbf{I}\right):\nabla \mathbf{v}&\nn\\
    - \sum_{\mA}\nabla \left( \rho^{-1}\check{\mu}_\mA + \rho_\mA^{-1}\check{\lambda} \right)\cdot \bH_\mA -\sum_{\mA} \left( \rho^{-1}\check{\mu}_\mA + \rho_\mA^{-1}\check{\lambda} \right) \zeta_\mA &\geq 0.
\end{align}

\bibliographystyle{unsrtnat}
\bibliography{references}

\end{document}